\documentclass[11pt]{article}
\usepackage{amssymb}
\usepackage{amsmath}
\usepackage{amsthm}
\usepackage{graphicx,color}

\theoremstyle{definition}
\newtheorem{defin}{Definition}[section]
\newtheorem{thm}[defin]{Theorem}

\newtheorem{theorem}[defin]{Theorem}
\newtheorem{rem}[defin]{Remark}
\newtheorem{ex}[defin]{Example}
\newtheorem{cor}[defin]{Corollary}
\newtheorem{conjecture}[defin]{Conjecture}
\newtheorem{lemma}[defin]{Lemma}
\newtheorem{prop}[defin]{Proposition}

\newcommand{\twomat}[4]{\left(\begin{array}{cc}#1&#2\\#3&#4\end{array}\right)}
\newcommand{\twovec}[2]{\left(\begin{array}{c}#1\\#2\end{array}\right)}

\def\hil{{\mathcal H}}

\def\A{{\mathcal A}}
\def\B{{\mathcal B}}

\def\cE{{\mathcal E}}
\def\cF{{\mathcal F}}

\def\S{{\mathcal S}}
\def\X{{\mathcal X}}
\def\U{{\mathcal U}}

\def\half{\frac{1}{2}}

\def\ep{\varepsilon}
\def\bN{\mathbb{N}}
\def\bC{\mathbb{C}}
\def\bR{\mathbb{R}}

\def\bz{\left(}
\def\jz{\right)}

\def\kii{\emph}
\def\kiii{}

\def\map{\Phi}

\def\I{\mathcal{I}}

\def\pls{\overline p_e}
\def\pli{\underline p_e}
\def\sa{_{\mathrm{sa}}}

\newcommand{\ki}{\emph}

\newcommand{\s}{\mbox{ }}
\newcommand{\ds}{\mbox{ }\mbox{ }}

\newcommand{\norm}[1]{\left\| #1\right\|}

\newcommand{\inner}[2]{\left\langle #1 , #2\right\rangle}

\newcommand{\pinner}[3]{\left\langle #1 | #2| #3\right\rangle}

\newcommand{\diad}[2]{|#1\rangle\langle #2|}
\newcommand{\pr}[1]{\diad{#1}{#1}}

\newcommand{\dmax}[2]{D_{\mathrm{max}}(#1\,\|\, #2)}
\newcommand{\ch}[2]{C(#1,#2)}
\newcommand{\ket}[1]{|#1\rangle}
\newcommand{\vecc}[1]{\vec{#1}}

\newcommand{\comp}[1]{\bar{#1}}

\newcommand{\be}{\begin{equation}}
\newcommand{\ee}{\end{equation}}
\newcommand{\bea}{\begin{eqnarray}}
\newcommand{\eea}{\end{eqnarray}}
\newcommand{\beas}{\begin{eqnarray*}}
\newcommand{\eeas}{\end{eqnarray*}}

\DeclareMathOperator{\id}{I}
\DeclareMathOperator{\Tr}{Tr}
\DeclareMathOperator{\trace}{Tr}

\DeclareMathOperator{\diag}{Diag}

\def\qopnamewl@#1{\mathop{\rm{#1}}}
\def\argmin{\qopnamewl@{arg\,min}}
\def\argmax{\qopnamewl@{arg\,max}}
\def\lub{\qopnamewl@{LUB}}
\def\LUB{\qopnamewl@{LUB}}
\def\glb{\qopnamewl@{GLB}}
\def\GLB{\qopnamewl@{GLB}}

\begin{document}

\pagestyle{plain}

\setlength{\topmargin}{-1.0cm}  
\setlength{\textheight}{8.6in} \setlength{\parskip}{1mm}

\setcounter{page}{1}

\title{Upper bounds on the error probabilities and asymptotic error exponents in quantum multiple state discrimination}
\author{Koenraad~M.R.~Audenaert$^{1,2,}$\footnote{Electronic address: koenraad.audenaert@rhul.ac.uk}\; and
Mil\'an Mosonyi$^{3,4,}$\footnote{Electronic address: milan.mosonyi@gmail.com}\\[3mm]
    \small\it $^1$Department of Mathematics, Royal Holloway University of London, \\
    \small\it Egham TW20 0EX, U.K.\\[1mm]
    \small\it $^2$Department of Physics and Astronomy, University of Ghent, \\
    \small\it S9, Krijgslaan 281, B-9000 Ghent, Belgium\\[1mm]
    \small\it $^3$F\'{\i}sica Te\`{o}rica: Informaci\'{o} i Fenomens Qu\`{a}ntics, Universitat Aut\`{o}noma de Barcelona, \\
    \small\it ES-08193 Bellaterra (Barcelona), Spain\\[1mm]
    \small\it $^4$Mathematical Institute, Budapest University of Technology and Economics \\
    \small\it Egry J\'ozsef u~1., Budapest, 1111 Hungary
}
\date{}
\maketitle

\begin{abstract}
We consider the multiple hypothesis testing problem for symmetric quantum state discrimination between $r$ given states $\sigma_1,\ldots,\sigma_r$.
By splitting up the overall test into multiple binary tests in various ways we obtain a number of upper bounds on the optimal error probability
in terms of the binary error probabilities.
These upper bounds allow us to deduce various bounds on the asymptotic error rate, for which it has been hypothesized that it is given by the
multi-hypothesis quantum Chernoff bound (or Chernoff divergence) $C(\sigma_1,\ldots,\sigma_r)$,
as recently introduced by Nussbaum and Szko\l a in analogy with Salikhov's classical multi-hypothesis Chernoff bound.
This quantity is defined as the minimum of the pairwise binary Chernoff divergences $\min_{j<k} C(\sigma_j,\sigma_k)$.
It was known already that the optimal asymptotic rate must lie between $C/3$ and $C$, and that for certain classes
of sets of states the bound is actually achieved.
It was known to be achieved, in particular, when the state pair that is closest together in Chernoff divergence
is more than 6 times closer than the next closest pair.
Our results improve on this in two ways. Firstly, we show that the optimal asymptotic rate must lie between $C/2$ and $C$.
Secondly, we show that the Chernoff bound is already achieved when the closest state pair
is more than 2 times closer than the next closest pair.
We also show that the Chernoff bound is achieved when at least $r-2$ of the states are pure, improving on a previous result by Nussbaum and Szko\l a.
Finally, we indicate a number of potential pathways along which a proof (or disproof) may eventually be found that the multi-hypothesis quantum
Chernoff bound is always achieved.
\end{abstract}

\newpage
\section{Introduction}\label{sec:intro}

Consider a communication scenario where a sender (say, Alice) wishes to send one of $r$ possible messages to a receiver (Bob). 
To achieve this goal, Alice has a device at her disposal that can prepare $r$ quantum states $\rho_1,\ldots,\rho_r$ 
from some state space $\S(\hil)$, one for each possible message, 
which she can then send through a quantum channel $\map$, resulting in the states
$\sigma_i=\map(\rho_i)$ at Bob's side.
Bob then has to make a quantum measurement to identify which message was sent.
His measurement is described by a set of positive semidefinite operators $E_1,\ldots,E_r$, one corresponding to each possible message, 
that form an incomplete POVM (positive operator-valued measure), i.e., they satisfy $E_1+\ldots+E_r\le I$. The operator $E_0:=I-(E_1+\ldots+E_r)$ 
corresponds to not making a decision on the identity of the received state. The probability of making an erroneous decision when the message $i$ 
was sent is then given by $\Tr\sigma_i(I-E_i)$. If we also assume that Alice sends each message $i$ with a certain probability $p_i$ then 
the best Bob can do is choose the POVM that minimizes the Bayesian error probability
\begin{equation*}
P_e\bz\{E_1,\ldots,E_r\}\jz:=\sum_{i=1}^r p_i\Tr\sigma_i(I-E_i)=\sum_{i=1}^r \Tr A_i(I-E_i),
\end{equation*}
where $A_i:=p_i\sigma_i$.

In the classical case, i.e., when the $\sigma_i$ are mutually commuting,
the optimal success probability is known to be reached by the so-called maximum likelihood measurement, and
the optimal success probability is given by $\Tr\max\{p_1\sigma_1,\ldots,p_r\sigma_r\}$, where the maximum is taken entrywise
in some basis that simultaneously diagonalizes all the $\sigma_i$.
In the general quantum case, no explicit expression is known for the optimal error probability, or for the measurement achieving it, unless $r=2$, 
in which case these optimal quantities are easy to find \cite{Helstrom,Holevo78}.
Moreover, it turns out to be impossible to extend the notion of maximum of a set of real numbers to maximum of a set of 
positive semidefinite operators on a Hilbert space while keeping all the properties of the former --
technically speaking, the positive semidefinite ordering does not induce a lattice structure -- and 
because of this a straightforward generalization of the classical results is not possible.
In Section \ref{sec:ML}, we define a generalized notion of maximum for a set of
self-adjoint operators, which we call the least upper bound (LUB). This notion reduces to the usual maximum in the classical case, and
the optimal success probability can be expressed as $\Tr\lub(p_1\sigma_1,\ldots,p_r\sigma_r)$ \cite{ykl},
giving a direct generalization of the classical expression.
We explore further properties of the least upper bound, and its dual, the greatest lower bound (GLB), in Appendix \ref{sec:LUB}.

An obvious way to reduce the error probability is to send the same message multiple times. For $n$ repetitions, the optimal error probability is given by
\begin{equation}\label{optimal error}
P_e^*\bz A_{1,n},\ldots,A_{r,n}\jz:=\min\left\{\sum_{i=1}^r \Tr A_{i,n}(I-E_i)\,:\,
\{E_1,\ldots,E_r\}\s\text{POVM on }\S(\hil^{\otimes n})\right\},
\end{equation}
where $A_{i,n}:=p_i\sigma_i^{\otimes n}$. These error probabilities are known to decay exponentially fast in the number of repetitions \cite{Aud,NSz,NS10},
and hence we are interested in the exponents (which are negative numbers)
\begin{align}
\pli\bz\vecc{A}_1,\ldots,\vecc{A}_r\jz&:=\liminf_{n\to\infty}\frac{1}{n}\log
P_e^*\bz A_{1,n},\ldots,A_{r,n}\jz \ds\ds\ds\text{and}\ds\ds\ds\label{li exponent}\\
\pls\bz\vecc{A}_1,\ldots,\vecc{A}_r\jz&:=\limsup_{n\to\infty}\frac{1}{n}\log
P_e^*\bz A_{1,n},\ldots,A_{r,n}\jz,\label{ls exponent}
\end{align}
where $\vecc{A_i}:=\{A_{i,n}\}_{n\in\bN}$.

In the case of two possible messages, the theorem for the quantum Chernoff bound \cite{Aud,ANSzV,NSz}
states that
 \begin{equation}\label{binary Chernoff}
 \pli\bz \vecc{A}_1,\vecc{A}_2\jz=
 \pls\bz \vecc{A}_1,\vecc{A}_2\jz=
\min_{0\le t\le 1}\log\Tr\sigma_1^t\sigma_2^{1-t}
=:
-\ch{\sigma_1}{\sigma_2},
 \end{equation}
where
$\ch{\sigma_1}{\sigma_2}$ is a positive quantity known as the \ki{Chernoff divergence} of $\sigma_1$ and $\sigma_2$.
According to a long-standing conjecture, it is hypothesized that
\begin{align}\label{main conjecture}
\pli\bz\vecc{A}_1,\ldots,\vecc{A}_r\jz=
\pls\bz\vecc{A}_1,\ldots,\vecc{A}_r\jz=
-\min_{(i,j):\,i\ne j}\ch{\sigma_i}{\sigma_j},
\end{align}
i.e., the multi-hypothesis exponent is equal to the worst-case pairwise exponent.
Following \cite{NS10}, we call $C(\sigma_1,\ldots,\sigma_r):=\min_{(i,j):\,i\ne j}\ch{\sigma_i}{\sigma_j}$ the \ki{multi-Chernoff bound}. 
In fact, the lower bound $\pli(\vecc{A}_1,\ldots,\vecc{A}_r)\ge -C(\sigma_1,\ldots,\sigma_r)$ (optimality) follows trivially from the 
binary case \cite{NSz}, as was pointed out e.g., in \cite{NS11a}.
The upper bound $\pls(\vecc{A}_1,\ldots,\vecc{A}_r)\le -C(\sigma_1,\ldots,\sigma_r)$ (achievability)
is known to be true for commuting states \cite{salikhov73} and when the states $\sigma_i$ have pairwise disjoint supports \cite{NS11b}. 
A special case of the latter is when all the states $\sigma_i$ are pure \cite{NS11a}.

Our aim here is to establish decoupling bounds on the \ki{single-shot} error probability
by decomposing a multi-hypothesis test into multiple binary tests. These bounds in turn yield bounds on the exponents 
\eqref{li exponent}--\eqref{ls exponent} in terms of the corresponding pairwise exponents.
We remark that the existing asymptotic results mentioned in the previous paragraph also rely implicitly on single-shot decoupling bounds. Regarding lower bounds,
it has been shown in \cite{Qiu} that, for any choice of $\sigma_1,\ldots,\sigma_r$, and priors $p_1,\ldots,p_r$, we have
\begin{align}\label{single-shot lower bounds}
P_e^*\bz A_{1},\ldots,A_{r}\jz\ge \frac{1}{r-1}\sum_{(i,j):\,i<j}P_e^*(A_i,A_j),
\end{align}
where $A_i:=p_i\sigma_i$. Taking then $r$ sequences of states $\sigma_{i,n},\,n\in\bN$, and $\vecc{A_i}:=\{p_i\sigma_{i,n}\}_{n\in\bN}$, we get
\begin{align}\label{asymptotic lower bounds}
\pli\bz\vecc{A}_1,\ldots,\vecc{A}_r\jz\ge\max_{(i,j):\,i\ne j}
\pli\bz\vecc{A}_i,\vecc{A}_j\jz.
\end{align}
Note that this is true for arbitrary sequences of states $\{\sigma_{i,n}\}_{n\in\bN}$, with no special assumption on the correlations. In the
i.i.d.~case the right-hand side (RHS) of \eqref{asymptotic lower bounds} is exactly $-\min_{(i,j):\,i\ne j}\ch{\sigma_i}{\sigma_j}$,
and we recover the optimality part of \eqref{main conjecture}.

Hence, in this paper we will focus on upper decoupling bounds.
Upper bounds on the optimal error in terms of the pairwise fidelities can easily be obtained from some results in \cite{barnumknill}:
\begin{align}\label{BK1}
P^*_{e}(A_1,\ldots,A_r)\le \half\sum_{(i,j):\,i\ne j} \sqrt{p_ip_j}F(\sigma_i,\sigma_j),
\end{align}
where $F(\sigma_i,\sigma_j):=\|\sigma_i^{1/2}\sigma_j^{1/2}\|_1$ is the fidelity.
We provide a short proof of this bound in Appendix \ref{sec:Tysonproofs}.
When all the $\sigma_i$ are of rank one, the above bound can be improved as \cite{HLS}
\begin{align}\label{HLS1}
P^*_{e}(A_1,\ldots,A_r)\le \half\sum_{(i,j):\,i\ne j}\frac{p_i^2+p_j^2}{p_ip_j} F^2(\sigma_i,\sigma_j).
\end{align}
Using the Fuchs--van de Graaf inequalities \cite{FvdG},
these bounds can easily be translated into bounds in terms of the pairwise error probabilities, and we obtain
the following converses to \eqref{single-shot lower bounds} and
\eqref{asymptotic lower bounds}:
\smallskip

\noindent \textbf{Single-shot upper decoupling bounds:}\ds For $A_i:=p_i\sigma_i$, $i=1,\ldots,r$,
\begin{align}\label{sh-main1}
P^*_{e}(A_1,\ldots,A_r)\le\sum_{(i,j):\,i\ne j}\sqrt{p_i+p_j}\sqrt{P_e^*(A_i,A_j)}.
\end{align}
If $A_i$ is rank one for all $i$ then the square root can be removed from the pairwise errors; more precisely,
\begin{align}\label{sh-main2}
P^*_{e}(A_1,\ldots,A_r)\le\sum_{(i,j):\,i\ne j}(p_i+p_j)\frac{p_i^2+p_j^2}{p_i^2p_j^2}P_e^*\bz A_i,A_j\jz.
\end{align}
\smallskip


These single-shot bounds immediately yield the following
\smallskip

\noindent \textbf{Asymptotic upper decoupling bounds:}\ds For $\vecc{A}_i:=\{p_i\sigma_{i,n}\}_{n\in\bN}$, $i=1,\ldots,r$,
\begin{align}\label{a-main1}
\pls\bz\vecc{A}_1,\ldots,\vecc{A}_r\jz\le\half\max_{(i,j):\,i\ne j}\pls(\vecc{A}_i,\vecc{A}_j).
\end{align}
If $A_{i,n}$ is rank one for all $i$ and $n$ then
\begin{align}\label{a-main2}
\pls\bz\vecc{A}_1,\ldots,\vecc{A}_r\jz\le\max_{(i,j):\,i\ne j}\pls(\vecc{A}_i,\vecc{A}_j).
\end{align}
\smallskip

Note that again \eqref{a-main1} and \eqref{a-main2} are true for \ki{arbitrary} sequences of states, and in the i.i.d.~case we have
$\max_{(i,j):\,i\ne j}\pls(\vecc{A}_i,\vecc{A}_j)=-\min_{(i,j):\,i\ne j}\ch{\sigma_i}{\sigma_j}=-C(\sigma_1,\ldots,\sigma_r)$.
In particular, by \eqref{asymptotic lower bounds} and \eqref{a-main2} we recover the result of \cite{NS11a}, i.e., that
\eqref{main conjecture} is true for pure states. 
For mixed states, Theorem 3 in \cite{NS11b} gives that $\pls(\vecc{A}_1,\ldots,\vecc{A}_r)$ is between $-C(\sigma_1,\ldots,\sigma_r)$ and
$-\frac{1}{3}C(\sigma_1,\ldots,\sigma_r)$.
Our bound \eqref{a-main1} improves the factor $1/3$ in this upper bound to $1/2$, which is the best general result known so far.

Analytical proofs for various special cases and extensive numerical simulations for the general case suggest that the square root in \eqref{sh-main1}
-- and, consequently, the factor $\half$ in \eqref{a-main1} -- can be removed.
The fidelity bounds in \eqref{BK1} and \eqref{HLS1} were obtained by bounding from above the error probability of the pretty good
measurement, using matrix analytic techniques. Here we explore a completely different approach to obtaining upper decoupling
bounds. Namely, we show that the optimal error $P_e^*$ can be bounded from above by the sum of the optimal error probabilities of $r$
binary state discrimination problems, where in each of these problems, the goal is to discriminate one of the original hypotheses
from all the rest. Thus, finding the optimal error exponent of the symmetric multiple state discrimination problem with i.i.d.~hypotheses 
can be reduced to finding the optimal error exponent of
the correlated binary state discrimination problem where one of the hypotheses is i.i.d., while the other is a convex mixture
of i.i.d.~states. This latter problem is interesting in its own right, as it is arguably the simplest non-i.i.d.~state discrimination problem 
and yet its solution is not yet known, despite considerable effort
towards establishing non-i.i.d.~analogs of the binary Chernoff bound theorem \cite{HMO,HMO2,HMH,MHOF,M}.
Here we make some progress towards the solution of this problem, and provide a complete solution when the i.i.d.~state is pure.

The structure of the paper is as follows. In Section \ref{sec:prelim} we summarize the necessary preliminaries and review the known results 
that are relevant for the rest of the paper. In particular, we give a short proof of \eqref{single-shot lower bounds}, and summarize the known 
results for the binary case. We also introduce the notion of the least upper bound for self-adjoint operators, and show how the optimal error 
probability can be expressed in this formalism.

In Section \ref{sec:Tyson} we first review the fidelity bounds of \cite{barnumknill} and \cite{HLS},
which are based on the performance of the suboptimal pretty good measurement.
Then we follow a similar approach to obtain bounds in terms of pairwise fidelity-like quantities.
From these bounds we can recover
\eqref{BK1}--\eqref{HLS1} up to a constant factor, and for some configurations they are strictly better than
\eqref{BK1}--\eqref{HLS1}.
This approach is based on Tyson's bound \cite{Tyson} on the performance of an other suboptimal family of measurements, the square
measurements.

In Section \ref{sec:averaged} we study a special binary problem where one of the hypotheses is i.i.d., while the other one is 
averaged i.i.d., i.e., a convex mixture of i.i.d.~states.
In the setting of Stein's lemma (with the averaged state being the null-hypothesis) the corresponding error exponent is known to be the worst-case pairwise exponent
\cite{BS} (see also \cite{M13} for a simple proof), and we conjecture the same to hold in the symmetric setting of the Chernoff bound.
Similarly to the case of multiple hypotheses, it is easy to show that the worst-case pairwise exponent cannot be exceeded (optimality).
In Theorem \ref{thm:averaged upper bounds} we present upper decoupling bounds on the error probability, analogous to \eqref{sh-main1} and 
\eqref{sh-main2}, which in the asymptotics yield that $1/2$ times the conjectured exponent is achievable. Moreover, when the i.i.d.~state is 
pure then the factor $1/2$ can be removed and we get both optimality and achievability.

In Section \ref{sec:dichotomic} we show that the exponential decay rate of the optimal error probability \eqref{optimal error} 
is the same as that of another quantity, which we call the \ki{dichotomic error}. This is defined as the sum of the error probabilities 
of the binary state discrimination problems where we only want to decide whether hypothesis $i$ is true or not, for every $i=1,\ldots,r$. 
In the i.i.d.~case these binary problems are exactly of the type discussed in Section \ref{sec:averaged}, and we can directly apply the 
bounds obtained there to get upper decoupling bounds on both the single-shot and the asymptotic error probabilities, which give the bounds 
\eqref{sh-main1}--\eqref{a-main2}.

In Section \ref{sec:nussbaum} we follow Nussbaum's approach \cite{N13} to obtain a different kind of decoupling of the optimal error
probability. When applied recursively and combined with the bounds of Section \ref{sec:averaged}, this approach provides an alternative way 
to obtain bounds of the type \eqref{sh-main1}--\eqref{sh-main2}, which again yield
\eqref{a-main1}--\eqref{a-main2} in the asymptotics.

In Section \ref{sec:asymptotics} we show how the various single-shot bounds of the above described approaches translate into
bounds for the error rates, i.e., we derive \eqref{a-main1}--\eqref{a-main2} for the most general scenario, and its variants
for more specific settings, where the pairwise rates can be replaced by pairwise Chernoff divergences.
In particular, we improve on the result of \cite{NS11a} by showing that \eqref{main conjecture} holds if at least $r-2$ of the states
$\sigma_i$ are pure.
We also give an improvement of Nussbaum's asymptotic result \cite{N13}, which says that \eqref{main conjecture} is true if there is a 
pair of states $(\sigma_i,\sigma_j)$ such that $C(\sigma_i,\sigma_j)<\frac{1}{6}C(\sigma_k,\sigma_l)$ for any
$(k,l)\ne(i,j)$. Here we show that the constant $\frac{1}{6}$ can be replaced with $\frac{1}{2}$.

Supplementary material is provided in a number of Appendices.
In Appendix \ref{sec:LUB}, we explore some properties of the least upper bound and the greatest lower bound for self-adjoint
operators, which further extend their analogy to the classical notions of minimum and maximum.
In Appendix \ref{sec:classical} we show how our approaches
work in the classical case (when all operators commute), thus providing various alternative proofs for \eqref{main conjecture} in the classical case.
In Appendix \ref{sec:pure} we review the pure state case; we show an elementary way to derive the Chernoff bound theorem \eqref{binary Chernoff} 
for two pure states, and show how the
combination of the single-shot bounds of \cite{HLS} and \cite{Qiu} yield \eqref{main conjecture} for an arbitrary number of pure states.
In Appendix \ref{sec:semidef} we review the dual formulation of the optimal error probability due to \cite{ykl}.
For readers' convenience, we provide a proof for Tyson's and Barnum and Knill's error bounds in Appendix \ref{sec:Tysonproofs}.

\section{Preliminaries}\label{sec:prelim}

\subsection{Notations}

For a finite-dimensional Hilbert space $\hil$, let $\B(\hil)$
denote the set of linear operators on $\hil$, let $\B(\hil)\sa$ denote the set of self-adjoint (Hermitian) operators,
$\B(\hil)_+$ the set of positive semidefinite (PSD) operators, and
$\S(\hil)$ the set of density operators (states), i.e., the set of PSD operators with unit trace.

For $X$ a Hermitian operator, let $|X|$ denote its absolute value (or modulus), $|X|:=\sqrt{X^2}$.
The Jordan decomposition of $X$ into its positive and negative parts is given by
$X=X_+ - X_-$, with $X_\pm = (|X|\pm X)/2$, and $|X|=X_++X_-$. It is clear that $X_+ X_-=0$.
As the eigenvalues of $|X|$ are the absolute values of the eigenvalues of $X$,
the eigenvalues of $X_+$ ($X_-$) are the positive (negative) eigenvalues of $X$.
We denote the projections onto the support of $X_+$ and $X_-$ by
$\{X>0\}$ and $\{X<0\}$, respectively.

We will follow the convention that powers of a positive semidefinite (PSD) operator are only taken on its support. 
That is, if $a_1,\ldots,a_r$ are the strictly positive eigenvalues of $A\ge 0$,
with corresponding spectral projections $P_1,\ldots,P_r$, then
$A^s:=\sum_{i=1}^r a_i^s P_i$ for every $s\in\bR$. In particular, $A^0$ denotes the projection onto the support of $A$.

By a POVM we will mean a set of PSD operators $E_1,\ldots,E_r$ such that
$E_1+\ldots+E_r\le I$.
On occasion we will also consider the underlying measurement operators $\{X_k\}_{k=1}^r$, which are sets of operators such that
the products $X_k^*X_k$ constitute a POVM.

We will normally not indicate the base of the logarithm, but we will always assume that it is larger than $1$, and hence $\log$ is a strictly increasing function.
We will use the conventions $\log 0:=-\infty$ and $\log+\infty:=+\infty$.

\subsection{The problem setting}

We will consider a generalized state discrimination problem, where the hypotheses are represented by
arbitrary non-zero PSD operators (i.e., not necessarily states). We consider such a generalized setting partly
to absorb the priors into the states to make the formalism simpler, and partly
because
the formalism supports it, and all our results can be formulated and proved in this more general setting.
More importantly, however, we need to treat such generalized setups even if we restrict our original hypotheses
to be states; see, e.g., Lemma \ref{lem:nussbaum}.

More in detail, in the single-shot case our hypotheses are represented by non-zero PSD operators
$A_1,\ldots,A_r$. Occasionally, we will use the notations
\begin{equation*}
p_k:=\Tr A_k,\ds\ds\ds\sigma_k:=A_k/p_k,\ds\ds k=1,\ldots,r.
\end{equation*}
If $p_1+\ldots+p_r=1$ then we say that $\{A_k\}$ forms a set of \ki{weighted states}.

For any POVM $E_1,\ldots,E_r$, we define the corresponding \ki{success- and error probabilities} as
\begin{align*}
P_s(\{E_i\}):=\sum_{i=1}^r\Tr A_i E_i,\ds\ds\ds\ds\ds\ds
P_e(\{E_i\}):=\sum_{i=1}^r\Tr A_i (I-E_i).
\end{align*}
These can indeed be interpreted as probabilities in the case of weighted states, whereas in the general case they might take values above $1$.
Since it will always be obvious what the hypotheses are, we don't indicate them in the above notations. The optimal values of these quantities 
over all possible choices of POVMs are the optimal success- and error probability
\begin{align}
P_s^*(A_1,\ldots,A_r)&:=\max\left\{ \sum\nolimits_{i=1}^r\Tr A_i E_i\,:\,\{E_1,\ldots,E_r\}\s\text{POVM}\right\},
\label{success prob}\\
P_e^*(A_1,\ldots,A_r)&:=\min\left\{\sum\nolimits_{i=1}^r\Tr A_i (I-E_i)\,:\,\{E_1,\ldots,E_r\}\s\text{POVM}\right\}.\nonumber
\end{align}
The maximum and the minimum above exist because the domain of optimization is compact and the functions to optimize are
continuous with respect to any natural topology on the set of POVMs on a fixed set of outcomes.
We will use the shorthand notations $P_s^*$ and $P_e^*$ when it is clear what the hypotheses are.
Note that
\begin{equation*}
P_s^*(A_1,\ldots,A_r)+P_e^*(A_1,\ldots,A_r)=\Tr A_0,\ds\ds
\text{where}\ds\ds A_0:=\sum_{i=1}^r A_i.
\end{equation*}
Again, these can be interpreted as probabilities if $1=\Tr A_0=\sum_i p_i$, i.e., in the case of weighted states.
Note that any POVM that is optimal for $P_s^*$ is also optimal for $P_e^*$ and vice versa. Moreover, there always exists an
optimal POVM $\{E_i\}_{i=1}^r$ such that $E_1+\ldots+E_r=I$.

In the asymptotic setting, our hypotheses are going to be represented by sequences of PSD operators, $\vecc{A_i}:=\{A_{i,n}\}_{n\in\bN},\,i=1,\ldots, r$, 
and we will be interested in the exponents $\pli$ and $\pls$, defined in \eqref{li exponent} and \eqref{ls exponent}, respectively.
We say that the $i$-th hypothesis is \ki{i.i.d.} (for the classical analogy of independent and identically distributed)
if $p_{i,n}=\Tr A_{i,n}$ is independent of $n$ (and hence we can define $p_i:=p_{i,n}$, $n\in\bN$),
and $A_{i,n}/p_i=\sigma_i^{\otimes n}$ for every $i$, where $\sigma_i:=\sigma_{i,1}$.
We say that the asymptotic state discrimination problem is i.i.d. if all the hypotheses are i.i.d.

When passing from single-shot error bounds to asymptotic error bounds, we will use the following standard lemma without further notice.
\begin{lemma}\label{lemma:maximum rate}
Let $a_{i,n},\,n\in\bN,\,i=1,\ldots,r$, be sequences of positive numbers. Then
\begin{align*}
\max_i\liminf_{n\to\infty}\frac{1}{n}\log a_{i,n}\le
\liminf_{n\to\infty}\frac{1}{n}\log\sum_{i=1}^r a_{i,n}
\le
\limsup_{n\to\infty}\frac{1}{n}\log\sum_{i=1}^r a_{i,n}
\le
\max_i\limsup_{n\to\infty}\frac{1}{n}\log a_{i,n}.
\end{align*}
\end{lemma}
\begin{proof}
The first inequality is straightforward from $a_{i,n}\le \sum_i a_{i,n},\,\forall i$, and the second inequality is obvious. To prove the last inequality, let
$M:=\max_i\limsup_{n\to\infty}\frac{1}{n}\log a_{i,n}$. If $M=+\infty$ then the assertion is trivial, and hence we assume that $M<+\infty$. 
By the definition of the limit superior, for every
$M'>M$, there exists an $N_{M'}$ such that for all $n\ge N_{M'}$, $a_{i,n}<\exp(M'),\,i=1,\ldots,r$, and hence $\frac{1}{n}\log\sum_{i=1}^r a_{i,n}<\frac{1}{n}\log r+M'$. Thus
$\limsup_{n\to\infty}\frac{1}{n}\log\sum_{i=1}^r a_{i,n}\le M'$. Since this is true for all $M'>M$, the assertion follows.
\end{proof}

\subsection{The generalized maximum likelihood error}
\label{sec:ML}

In the classical state discrimination problem, where the hypotheses are represented by non-negative functions
$A_i:\,\X\to\bR_+$ on some finite set $\X$, the optimal success probability is known to be
$\sum_x\max\{A_1,\ldots,A_r\}$, and it is achieved by the maximum likelihood measurement (see Appendix \ref{sec:classical} for details). 
If we consider the $A_i$ as diagonal operators in some fixed basis, then $P_s^*$
can be rewritten as
\begin{align}\label{classical success}
P_s^*(A_1,\ldots,A_r)=\Tr\max\{A_1,\ldots,A_r\},
\end{align}
where $\max\{A_1,\ldots,A_r\}$ is the operator with $\max_i A_i(x)$ in its diagonals.
Note that this is not a maximum in the usual sense of PSD ordering;
indeed, it is well-known that the PSD ordering does not induce a lattice structure \cite{ando99}, so in general
the set of upper bounds to $r$ given self-adjoint operators $A_1,\ldots,A_r$, which is defined as $\A:=\{Y: Y\ge A_k,\, k=1,\ldots,r\}$,
has no minimal element, not even when the $A_k$ mutually commute; see, e.g. Example \ref{ex:no max} in Appendix \ref{sec:LUB}.
However, there is a unique minimal element within $\A$ in terms of the \emph{trace ordering}.
We can therefore define a least upper bound in this more restrictive sense as
\be
\lub(A_1,\ldots,A_r) := \argmin_{Y}\{\trace Y: Y\ge A_k,\, k=1,\ldots,r\}.\label{eq:defLUB}
\ee
For the proof of uniqueness, see Appendix \ref{sec:LUB}.
In a similar vein we can define the \textit{greatest lower bound} (GLB) as
\begin{equation}
\glb(A_1,\ldots,A_r) := \argmax_Y\{\trace Y: Y\le A_k,\, k=1,\ldots,r\}.\label{eq:defGLB}
\end{equation}
Clearly, we have
\begin{equation}
\glb(A_1,\ldots,A_r)  = -\lub(-A_1,\ldots,-A_r).\label{eq:LUBGLB}
\end{equation}
For further properties of the above notions, see Appendix \ref{sec:LUB}.

Note that the set of $k$-outcome POVMs forms a convex set, and the optimal success probability in \eqref{success prob} is given as the maximum 
of a linear functional over this convex set.
It was shown in \cite{ykl} that the duality of convex optimization yields
\begin{align}\label{dual formulation}
P_s^*(A_1,\ldots,A_r)=\min\{\Tr Y: Y\ge A_k,\, k=1,\ldots,r\}
\end{align}
(see also \cite{KRS} for a different formulation of the same result). Using the definition of the LUB above, this can be rewritten as
\begin{align}\label{ml success}
P_s^*(A_1,\ldots,A_r)=\Tr\lub(A_1,\ldots,A_r),
\end{align}
in complete analogy with the classical case \eqref{classical success}. For readers' convenience, we provide a detailed derivation of 
\eqref{dual formulation} in Appendix \ref{sec:semidef}.

For an ensemble of PSD operators $\{A_i\}_{i=1}^r$,
we define the \textit{complementary operator} of $A_i$ as the
operator given by the sum of all other  operators in the ensemble:
\begin{align*}
\bar A _i:=\sum_{j\neq i}A_j=A_0-A_i,
\end{align*}
where $A_0=\sum_i A_i$.
The optimal error probability can be expressed in terms of the GLB of the
complementary density operators:
\begin{equation}\label{ml error}
P_e^*(A_1,\ldots,A_r) = \trace\glb(\bar A_1,\ldots,\bar A_r).
\end{equation}
This is easy to show:
\begin{align*}
P_e^* 
&=\min_{\{E_k\}} \sum_k \trace A_k \sum_{l:\,l\neq k} E_l
=\min_{\{E_k\}} \sum_l \trace E_l \sum_{k:\,k\neq l} A_k
=\min_{\{E_k\}} \sum_l \trace E_l \bar A_l
=-\max_{\{E_k\}} \sum_l \trace E_l(- \bar A_l)\\
&=
-\Tr\lub(-\bar A_1,\ldots,-\bar A_r)
=
\Tr\glb(\bar A_1,\ldots,\bar A_r),
\end{align*}
where we used \eqref{ml success}, \eqref{eq:LUBGLB}, and
that an optimal POVM can be chosen so that $E_1+\ldots+E_r=I$.
Note that this is in general different from $\trace\glb(A_1,\ldots,A_r)$, which is the minimal
i.e.\ worst-case success probability $P_{s,\min}=\min_{\{E_k\}} \sum_k \trace A_k E_k$.

In the binary case, i.e., when $r=2$,
we have
\begin{align*}
P_s^*(A_1,A_2)&=
\max\{\Tr A_1E+\Tr A_2(I-E):\,0\le E\le I\}
=
\Tr A_2+\max_{0\le E\le I}\Tr(A_1-A_2)E\\
&=
\Tr A_2+\Tr(A_1-A_2)_+
=
\half\Tr(A_1+A_2)+\half\norm{A_1-A_2}_1,
\end{align*}
and the maximum is attained at $E=\{A_1-A_2>0\}$;
this is the so-called Holevo-Helstr\"om measurement \cite{Holevo78,Helstrom}.
Consequently, we have
\begin{align}
P_e^*(A_1,A_2)&=
\min\{\Tr A_1(I-E)+\Tr A_2E:\,0\le E\le I\}\label{minimum error2}\\
&=
\Tr(A_1+A_2)-P_s^*(A,A_2)=\half\Tr(A_1+A_2)-\half\norm{A_1-A_2}.\label{minimum error}
\end{align}

Comparing these with \eqref{ml success} and \eqref{ml error}, and noting that in the binary
case $\comp{A}_1=A_2,\,\comp{A}_2=A_1$, we obtain
\begin{align*}
\Tr\lub(A_1,A_2)&=\half\Tr(A_1+A_2)+\half\norm{A_1-A_2}_1,\\
\Tr\glb(A_1,A_2)&=\half\Tr(A_1+A_2)-\half\norm{A_1-A_2}_1.
\end{align*}
For a more straightforward way to derive these identities, see Appendix \ref{sec:LUB}.

In the rest of the paper, we will use the notations $P_e^*(A,B),\,\Tr\glb(A,B)$ and
$\half\Tr(A+B)-\half\norm{A-B}_1$ interchangeably for PSD operators $A,B$.

\subsection{Chernoff bound for binary state discrimination}\label{sec:chernoff}

For PSD operators $A,B$ on the same Hilbert space, define
\begin{align}
Q_s(A\|B) &:= \trace A^sB^{1-s},\ds\ds\ds s\in\bR,\nonumber\\
Q_{\min}(A,B) &:= \min_{0\le s\le 1} Q_s(A\|B), 
\\
\ch{A}{B} &:= -\log Q_{\min}(A,B).\nonumber
\end{align}
The last quantity, $\ch{A}{B}$ is the \ki{Chernoff divergence} of $A$ and $B$.
As it was shown in Theorem 1 in \cite{Aud} (see also \cite{ANSzV,sandwich}),
\begin{align}\label{Aud ineq}
\half\Tr(A+B)-\half\norm{A-B}_1\le Q_s(A\|B),\ds\ds\ds s\in[0,1].
\end{align}

Consider now the generalized asymptotic binary hypothesis testing problem with hypotheses
$\vecc{A}_1,\vecc{A_2}$. By \eqref{Aud ineq}, we have
\begin{equation}\label{Aud error bound}
P_e^*(A_{1,n},A_{2,n})
=
\half\Tr(A_{1,n}+A_{2,n})-\half\norm{A_{1,n}-A_{2,n}}_1
\le Q_s(A_{1,n}\|A_{2,n})
\end{equation}
for every $n\in\bN$ and $s\in[0,1]$, and hence,
\begin{equation}\label{Chernoff upper}
\pls(\vecc{A}_1,\vecc{A_2})=
\limsup_{n\to\infty}\frac{1}{n}\log P_e^*(A_{1,n},A_{2,n})\le
-C(\vecc{A}_1,\vecc{A_2}),
\end{equation}
where
\begin{equation}\label{regularized Chernoff}
C(\vecc{A}_1,\vecc{A_2}):=
-\limsup_{n\to\infty}\frac{1}{n}\log Q_{\min}(A_{1,n},A_{2,n})
=
\liminf_{n\to\infty}\frac{1}{n}\ch{A_{1,n}}{A_{2,n}}
\end{equation}
is the \ki{regularized Chernoff divergence}.

In the i.i.d.~case, i.e., when $A_{1,n}=p_1\sigma_1^{\otimes n}$ and
$A_{2,n}=p_2\sigma_2^{\otimes n}$ for every $n\in\bN$, we have
$Q_s(A_{1,n}\|A_{2,n})=p_1^sp_2^{1-s}Q_s(\sigma_1\|\sigma_2)^n\le \max\{p_1,p_2\}Q_s(\sigma_1\|\sigma_2)^n$, and \eqref{Aud error bound} yields
\begin{equation*}
P_e^*(p_1\sigma_1^{\otimes n},p_2\sigma_2^{\otimes n})
\le
\max\{p_1,p_2\}Q_{\min}(\sigma_1,\sigma_2)^n=
\max\{p_1,p_2\}\exp\bz-n\ch{\sigma_1}{\sigma_2}\jz
\end{equation*}
for every $n\in\bN$. In particular,
\begin{equation*}
\limsup_{n\to\infty}\frac{1}{n}\log P_e^*(p_1\sigma_1^{\otimes n},p_2\sigma_2^{\otimes n})\le
-\ch{\sigma_1}{\sigma_2}.
\end{equation*}
(Note that in this case $C(\vecc{A}_1,\vecc{A_2})=\ch{\sigma_1}{\sigma_2}$).

The above argument shows that the asymptotic Chernoff divergence (which is equal to the
single-shot Chernoff divergence in the i.i.d.~case) is an achievable error rate. Optimality
means that no faster exponential decay of the optimal error is possible, i.e., that
\begin{equation*}
\pli(\vecc{A}_1,\vecc{A_2})=
\liminf_{n\to\infty}\frac{1}{n}\log P_e^*(A_{1,n},A_{2,n})\ge
-C(\vecc{A}_1,\vecc{A_2}).
\end{equation*}
This was shown to be true in the i.i.d.~case in \cite{NSz}.
Optimality for various correlated scenarios was obtained in \cite{HMO,HMO2,HMH,MHOF,M};
the classes of states covered include Gibbs states of finite-range translation-invariant interactions on a spin chain, and 
thermal states of non-interacting bosonic and fermionic lattice systems.

\subsection{Pairwise discrimination and lower decoupling bounds}
\label{sec:pair}

Consider the generalized state discrimination problem with hypotheses $A_1,\ldots,A_r$.
In this section we 
review a lower bound on the optimal error probability for discriminating between $r$ given states
in terms of the optimal \textit{pairwise} error probabilities, originally given in \cite{Qiu}.
Let us thereto define the following quantities:
\begin{eqnarray}
P_{s,2}^* \s:= \s
P_{s,2}^*(A_1,\ldots,A_r)
&:=& \frac{1}{r-1} \sum_{(k,l):\,k< l} P_s^*(A_k,A_l) = \frac{1}{r-1} \sum_{(k,l):\,k< l} \trace\lub(A_k,A_l) \nonumber \\
&=& \frac{1}{r-1} \sum_{(k,l):\,k< l} \half\bz \trace(A_k+A_l)+\norm{A_k-A_l}_1\jz.\label{eq:defps2}
\end{eqnarray}
and
\begin{eqnarray}
P_{e,2}^* \s :=\s
P_{e,2}^*(A_1,\ldots,A_r) &:=& \frac{1}{r-1} \sum_{(k,l):\,k< l} P_e^*(A_k,A_l) = \frac{1}{r-1} \sum_{(k,l):\,k< l} \trace\glb(A_k,A_l) \nonumber \\
&=& \frac{1}{r-1} \sum_{(k,l):\,k< l} \half\bz \trace(A_k+A_l)-\norm{A_k-A_l}_1\jz.\label{eq:defpe2}
\end{eqnarray}
Note that
\begin{equation*}
P_{s,2}^*(A_1,\ldots,A_r)+P_{e,2}^*(A_1,\ldots,A_r)=\sum_i\Tr A_i=P_{s}^*(A_1,\ldots,A_r)+P_{e}^*(A_1,\ldots,A_r),
\end{equation*}
explaining the choice $1/(r-1)$ for the normalization.

In the case of weighted states, i.e., when $\Tr A_0=1$,
we can interpret these quantities as optimal success and error probabilities in a very special setting, whereby the receiver
can make use of a particular kind of side information. We shall assume that this side information has been provided by an oracle.
The oracle knows the correct value of each symbol sent out by the source, but in the best of oracular traditions, does not quite
reveal this information to the receiver.
Rather, the oracle provides the receiver with a choice of two symbols,
one of which is the correct one and the other
is chosen from the remaining values at random, with uniform probability $1/(r-1)$.
It is intuitively plausible that the receiver should only try to discriminate between the two options provided.

The optimal success probability in this setup can easily be calculated.
From the receiver's viewpoint, the probability that the values of the symbols provided by the oracle are
$k$ and $l$ (with $k\neq l$) is $p_k(r-1)^{-1}+(r-1)^{-1}p_l$, and the conditional probability
that $k$ is the correct one is $p_k/(p_k+p_l)$.
Hence, the receiver's optimal success probability will be
$$
P_{s,2}^*= \sum_{(k,l):\,k< l} \frac{p_k+p_l}{r-1}\;\;\;
\half\bz 1+\norm{\frac{p_k}{p_k+p_l}\sigma_k-\frac{p_l}{p_k+p_l}\sigma_l}_1\jz,
$$
which simplifies to (\ref{eq:defps2}).

It is intuitively clear that this oracle-assisted success probability should never be smaller than the unassisted
optimal success probability, whereby the receiver needs to discriminate between all $r$ possible symbols.
The following Theorem, first given in \cite{Qiu}, shows that this is indeed the case.
Here we give a detailed and slightly simplified proof for readers' convenience. We also provide a different proof and a strengthening of \eqref{ineq:Qiu}
in Theorem \ref{th:dich}.

\begin{theorem}\label{th:41}
For any $A_1,\ldots,A_r\in\B(\hil)_+$,
\begin{equation}\label{ineq:Qiu}
P_s^*(A_1,\ldots,A_r) \le P_{s,2}^*(A_1,\ldots,A_r),\ds\ds\mbox{ and } \ds\ds P_e^*(A_1,\ldots,A_r)\ge P_{e,2}^*(A_1,\ldots,A_r).
\end{equation}
\end{theorem}

\begin{proof}
First notice that
\begin{align*}
\sum_{(k,l):\,k\ne l}\bz \Tr A_k E_k+\Tr A_l E_l\jz=
2(r-1)\sum_{k=1}^r \Tr A_k E_k
\end{align*}
and
\begin{align*}
\Tr A_k E_k+ \Tr A_l E_l
&\le \Tr A_k E_k+\Tr A_l\bz I-E_k\jz
=   \trace A_l+\Tr\bz A_k-A_l\jz E_k\\
&\le \trace A_l+\Tr\bz A_k-A_l\jz_+
= \half(\trace(A_k+A_l)+\norm{A_k-A_l}_1).
\end{align*}
Hence,
\begin{align*}
\sum_{k=1}^r \Tr A_k E_k&=
\frac{1}{2(r-1)}\sum_{(k,l):\,k\ne l}\bz \Tr A_k E_k+ \Tr A_l E_l\jz \\
&\le
\frac{1}{4(r-1)}\sum_{(k,l):\,k\ne l}\bz \trace(A_k+A_l)+\norm{A_k-A_l}_1\jz,
\end{align*}
which yields the first assertion.
The second assertion, $P_e^*\ge P_{e,2}^*$, is now obvious.
\end{proof}

We conjecture that for any choice of signal states and source probabilities
the oracle-assisted error probability can not be arbitrarily smaller than the unassisted one.
In particular, we believe:
\begin{conjecture}\label{claim:1}
There exists a constant $c$, only depending on the number of hypotheses $r$, such that for all $A_1,\ldots,A_r\ge0$
\begin{equation*}
P_e^*(A_1,\ldots,A_r) \le c P_{e,2}^*(A_1,\ldots,A_r).
\end{equation*}
\end{conjecture}

We have ample numerical evidence for this conjecture, and this evidence suggests that $c=4(r-1)$.
Several approaches towards a proof will be provided in the next sections.

\subsection{Inequalities for various operator distinguishability measures}

We will often benefit from inequalities between various operator distinguishability measures. In particular, 
we will use inequalities between the optimal binary error, the Chernoff divergence, and the fidelity.
For positive semidefinite operators $A,B\in\B(\hil)_+$, their \ki{fidelity} is defined as
\begin{equation}\label{fidelity def}
F(A,B):=\norm{A^{1/2}B^{1/2}}_1=\trace(A^{1/2}BA^{1/2})^{1/2}.
\end{equation}

The following bounds between the fidelity and the trace-norm were shown in \cite{FvdG} for
states, and extended to weighted states in \cite{sandwich}, where also the sharpness
of the inequalities $\half\Tr(A+B)-\half\norm{A-B}_1\le F(A,B)\le
\sqrt{\frac{1}{4}\bz\Tr (A+B)\jz^2-\frac{1}{4}\norm{A-B}_1^2}$
has been shown. The proof for the general case can be obtained exactly the same way as in the
above cases.
\begin{lemma}\label{lemma:FT bounds}
For any $A,B\in\B(\hil)_+$,
\begin{align}
P_e^*(A,B)&=\Tr\glb(A,B)=\half\Tr(A+B)-\half\norm{A-B}_1\le F(A,B)\label{fid lower}\\
&\le
\sqrt{\frac{1}{4}\bz\Tr (A+B)\jz^2-\frac{1}{4}\norm{A-B}_1^2}
=
\sqrt{\Tr\lub(A,B)}\sqrt{\Tr\glb(A,B)}\nonumber\\
&\le
\sqrt{\Tr(A+B)}\sqrt{\Tr\glb(A,B)}
=
\sqrt{\Tr(A+B)}\sqrt{P_e^*(A,B)}.
\nonumber
\end{align}
\end{lemma}

When $A$ is rank one, we also have the following inequality. This has been stated as an exercise in \cite{NC} for states; 
we provide a proof here for readers' convenience.
\begin{lemma}\label{lem:FTpure}
Let $A,B\in\B(\hil)_+$ and assume that $A$ has rank one. Then
\begin{align}\label{F^2 bound}
P_e^*(A,B)=\Tr\glb(A,B)&=\half\Tr(A+B)-\half\norm{A-B}_1
\le\frac{1}{\Tr A} F(A,B)^2.
\end{align}
\end{lemma}
\begin{proof}
Let $\tilde A:=A/\Tr A$.
The assumption that $A$ is rank one yields that
$\Tr A\tilde A=\Tr A$ and $F(A,B)=\sqrt{\Tr AB}$.
Using the representation \eqref{minimum error}, we get
\begin{align*}
\half\Tr(A+B)-\half\norm{A-B}_1&=
\min\{A(I-E)+\Tr BE:\,0\le E\le I\}\\
&\le
\Tr A(I-\tilde A)+\Tr B\tilde A
=
\frac{1}{\Tr A}\Tr AB
=
\frac{1}{\Tr A}F(A,B)^2.\qedhere
\end{align*}
\end{proof}

\begin{rem}
Monotonicity of the fidelity under the trace yields that
$F(A,B)\le(\Tr A)^{1/2}(\Tr B)^{1/2}$. If $\Tr B\le\Tr A$ then
$F(A,B)\le\Tr A$, or equivalently, $\frac{1}{\Tr A}F(A,B)^2\le F(A,B)$, and hence the upper
bound in \eqref{F^2 bound} is stronger than the inequality in \eqref{fid lower}.
This is the case, for instance, for states.
In general, however, the two bounds are not comparable.
\end{rem}

According to Theorem 6 in \cite{ANSzV}, for any $A,B\in\B(\hil)_+$,
\begin{align*}
F(A,B)^2\le \Tr A^t B^{1-t}(\Tr A)^{1-t}(\Tr B)^t,\ds\ds\ds t\in[0,1].
\end{align*}
In particular, for states $\rho,\sigma$,
\begin{align}\label{fid chernoff}
F(\rho,\sigma)^2\le Q_{\min}(\rho,\sigma)=\min_{0\le t\le 1}\Tr\rho^t\sigma^{1-t}.
\end{align}

\section{Upper bounds from suboptimal measurements}
\label{sec:Tyson}

Consider the generalized state discrimination problem with hypotheses
$A_i\in\B(\hil)_+,\,i=1,\ldots,r$. As before, we write $A_i=p_i\sigma_i$, with $\Tr\sigma_i=1$.
When the number of hypotheses $r$ is larger than $2$, there is no explicit expression known
for the optimal error probability $P_e^*(A_1,\ldots,A_r)$ in general.
Obviously, any measurement yields an upper bound on the optimal error probability, some of which are known to have the same 
asymptotics in the limit of infinitely many copies as the optimal error probability. Here we first review the pretty good 
measurement (PGM), and the bounds \eqref{BK1}--\eqref{HLS1} from \cite{barnumknill,HLS}.
Next, we consider the
square measurement (SM), and derive upper bounds on its optimal error probability.
These upper bounds sometimes outperform those of \eqref{BK1}--\eqref{HLS1}.

For every $\alpha\in\bR$, define the $\alpha$-weighted POVM $\cE^{(\alpha)}$ by
\begin{equation*}
E_k^{(\alpha)}:=S_{\alpha}^{-1/2}A_k^{\alpha}S_{\alpha}^{-1/2},\ds\ds\ds
S_{\alpha}:=\sum_k A_k^{\alpha}.
\end{equation*}
Note that if $A_k=\pr{\psi_k}$ for some vectors $\psi_k$ then
$E_k^{(\alpha)}=\pr{\psi_k^{(\alpha)}}$ with $\psi_k^{(\alpha)}:=S_{\alpha}^{-1/2}\norm{\psi_k}^{\alpha-1}\psi_k$, and
$\sum_k\pr{\psi_k^{(\alpha)}}=\bz\sum_k A_k\jz^0$. In particular, if $\psi_1,\ldots,\psi_r$ are linearly independent 
then $\psi_1^{(\alpha)},\ldots,\psi_r^{(\alpha)}$ is an orthonormal system, spanning the same subspace as the original $\psi$ vectors.
That is, the above procedure yields an orthogonalization of the original set of vectors, which is different from the 
Gram-Schmidt orthogonalization in general.

The case $\alpha=1$ yields the so-called \textit{pretty good measurement} (PG) \cite{PGM}.
Barnum and Knill \cite{barnumknill} have shown that in the case of weighted states, the success probability of the PG measurement is bounded
below by the square of the optimal success probability:
$(P_s^*)^2 \le P_s^{PG} \le P_s^*$, which in turn yields that
$\half P_e^{PG}\le P_e^*\le P_e^{PG}$. In particular, $P_e^*$ and $P_e^{PG}$ have the same exponential decay rate in the
asymptotic setting. 
\begin{thm}[Barnum and Knill]\label{th:BKbound}
\begin{align}\label{BK bound}
P_e^{PG}\le \half\sum_{(i,j):\,i\ne j}F(A_i,A_j)=
\half\sum_{(i,j):\,i\ne j}\sqrt{p_ip_j}F(\sigma_i,\sigma_j).
\end{align}
\end{thm}
Actually, Theorem 4 in \cite{barnumknill} gives the upper bound in
\eqref{BK bound} without the $1/2$ pre-factor. 
We give a short proof of the improved bound in Appendix \ref{sec:Tysonproofs}.
This theorem immediately yields \eqref{BK1}.
It was shown in \cite{HLS} that when all the $A_i$ are rank one then
\begin{align}\label{HLS bound}
P_e^{PG}\le
\half\sum_{(i,j):\,i\ne j}\frac{p_i^2+p_j^2}{p_i^2p_j^2} F(A_i,A_j)^2=
\half\sum_{(i,j):\,i\ne j}\frac{p_i^2+p_j^2}{p_ip_j} F(\sigma_i,\sigma_j)^2,
\end{align}
which yields \eqref{HLS1}.
\medskip

The case $\alpha=2$ yields the
\textit{square measurement} (SQ),
with POVM elements
$$
E_{SQ;k}=X_{SQ;k}^*X_{SQ;k},\ds\ds\ds
X_{SQ;k}=A_k \bz\sum_k A_k^2\jz^{-1/2}.
$$
This type of measurement
has been used by various authors \cite{BM,Curlander,jezek,Holevo78}, and it
features in Tyson's bounds on the error probability \cite{Tyson}, which we briefly review below.
For a comprehensive overview of the use of the pretty good and the square measurements for state discrimination, see \cite{Tyson2}.

For any set $\{X_k\}$ of measurement operators (i.e., $\sum_k X_k^*X_k\le I$), let
\be
\Gamma(\{X_k\}) := \Tr A_0-\sum_{k=1}^r \norm{X_k A_k}_1.
\ee
Minimizing over all possible choices of $\{X_k\}$ yields the optimal value $\Gamma^*$:
\be
\Gamma^*:=\inf_{\{X_k\}: \sum_k X_k^*X_k=I} \Gamma(\{X_k\}).
\ee
The importance of this quantity $\Gamma$ comes from a combination of two facts.
First, it differs from the error probability $P_e$
only by a factor between 1 and at most 2. Hence, $\Gamma^*$ is a good approximation of $P_e^*$, especially in the asymptotic regime.
Moreover, unlike the optimal error probability, $\Gamma^*$ can be calculated explicitly by a closed-form expression.
This is the content of the following two theorems, first proven by Tyson \cite{Tyson}.
For completeness, we provide short proofs in Appendix \ref{sec:Tysonproofs}.

\begin{theorem}[Tyson]\label{th:Tyson}
Let $A_1,\ldots,A_r\in\B(\hil)_+$ and
$\{E_k=X_k^*X_k\}$ be a POVM. Then
$$
\Gamma(\{X_k\}) \le P_e(\{E_k\}) \le 2\Gamma(\{X_k\}).
$$
In particular, for the optimal POVM and optimal $X_k$ that achieve the minimum:
\begin{equation}\label{Tyson bound}
\Gamma^* \le P_e^* \le 2\Gamma^*.
\end{equation}
\end{theorem}

\begin{theorem}[Tyson]\label{th:Tyson2}
Let $A_1,\ldots,A_r\in\B(\hil)_+$ and $A_0:=\sum_i A_i$. Then
\begin{equation}\label{T2}
\Gamma^* = \Tr A_0-\Tr\left(\sum_{i=1}^r A_i^2\right)^{1/2},
\end{equation}
with the optimal measurement operators being those of the SQ measurement.
\end{theorem}
From (\ref{T2}) it follows that $\Gamma^*$ can take values between $0$ (when all $A_i$ are mutually orthogonal)
and $\Tr A_0-\Tr A_0^2/\sqrt{r}$ (when all $A_i$ are equal), whereas $P_e^*$ lies between 0 and $(\Tr A_0)(1-1/r)$.

Tyson's theorems yield that
\begin{align}\label{square bound}
P_e^*\le P_e^{SQ}\le 2\left[ \Tr A_0-\Tr\left(\sum_{i=1}^r A_i^2\right)^{1/2}\right].
\end{align}
Thus any decoupling bound on the RHS of \eqref{square bound} yields a decoupling bound on $P_e^*$. Here we show the following:

\begin{prop}\label{prop:half}
Let $A_1,\ldots,A_r\in B(\hil)_+$ and $A_0:=\sum_i A_i$. Then
\begin{align*}
\Tr A_0-\trace\left(\sum_j A_j^{2}\right)^{1/2}
&\le
\Tr A_0-(\Tr A_0)^{\frac{3}{2}}\left((\Tr A_0)+2\sum_{(i,j):\,i<j} \trace A_i^{1/2}A_j^{1/2}\right)^{-1/2}\\
&\le
\sum_{(i,j):\,i<j} \trace A_i^{1/2}A_j^{1/2}.
\end{align*}
\end{prop}
\begin{proof}
According to Lieb's theorem, the functional $(B,C)\mapsto\trace(B^t C^{1-t})$ is jointly concave for $0<t\le 1$.
That is, for PSD operators $B_j$ and $C_j$,
$$
\trace \sum\nolimits_j B_j^t C_j^{1-t} \le \trace\bz\sum\nolimits_j B_j\jz^t\bz\sum\nolimits_j C_j\jz^{1-t}.
$$
Then, using the fact $|\trace X|\le ||X||_1$ and H\"older's inequality,
\begin{align*}
\trace \sum\nolimits_j B_j^t C_j^{1-t}
&\le
\norm{\bz\sum\nolimits_j B_j\jz^t\bz\sum\nolimits_j C_j\jz^{1-t}}_1
\le \norm{\bz\sum\nolimits_j B_j\jz^t}_{1/s}\;\;\norm{\bz\sum\nolimits_j C_j\jz^{1-t}}_{1/(1-s)}
\end{align*}
for every $0<s<1$.
Now take $t=2/3$, $s=1/3$,
$B_j=A_j^{1/2}$ and $C_j=A_j^2$, then
\begin{align*}
\trace \sum\nolimits_j A_j^{1/3} A_j^{2/3} &\le
\norm{\bz\sum\nolimits_j A_j^{1/2}\jz^{2/3}}_{3}\;\;\norm{\bz\sum\nolimits_j A_j^2\jz^{1/3}}_{3/2}\\
&=\bz\trace\bz\sum\nolimits_j A_j^{1/2}\jz^2\jz^{1/3}\;\;
\left(\trace\left(\sum\nolimits_j A_j^{2}\right)^{1/2}\right)^{2/3}.
\end{align*}
Obviously, the LHS equals $\trace A_0$.
Taking the $3/2$ power and rearranging then yields
\begin{align*}
\trace\left(\sum\nolimits_j A_j^{2}\right)^{1/2}
&\ge
(\Tr A_0)^{\frac{3}{2}}\left(\trace\left(\sum\nolimits_j A_j^{1/2}\right)^2\right)^{-1/2}
=
(\Tr A_0)^{\frac{3}{2}}\left(\sum\nolimits_{i,j} \trace A_i^{1/2}A_j^{1/2}\right)^{-1/2}\\
&=
(\Tr A_0)^{\frac{3}{2}}\left(\Tr A_0+2\sum_{(i,j):\,i<j} \trace A_i^{1/2}A_j^{1/2}\right)^{-1/2}\\
&\ge
\Tr A_0-\sum_{(i,j):\,i<j} \trace A_i^{1/2}A_j^{1/2},
\end{align*}
where in the last line we exploited the inequality $(a+x)^{-1/2}\ge a^{-1/2}-\half a^{-3/2} x$.
\end{proof}

\begin{theorem}\label{th:half}
Let $A_1,\ldots,A_r\in\B(\hil)_+$ and let $p_k:=\Tr A_k,\,\sigma_k:=A_k/\Tr A_k$. Then
\begin{align}
P_e^*(A_1,\ldots,A_r)\le
\sum_{(i,j):\,i\ne j}\Tr A_i^{1/2} A_j^{1/2}
&\le
\sum_{(i,j):\,i\ne j}F( A_i, A_j)
\label{eq:mixedboundT}\\
&\le
\begin{cases}
\sum_{(i,j):\,i\ne j}\sqrt{p_i+ p_j}\;\;P_e^*(A_i,A_j)^{1/2},\\
\sum_{(i,j):\,i\ne j} \sqrt{p_i p_j} \;\;Q_{\min}(\sigma_i,\sigma_j)^{1/2}
\end{cases}.\label{eq:mixedboundT2}
\end{align}
In the special case that all states $\sigma_k$ are pure, we have the improved bound
\begin{align}
P_e^*(A_1,\ldots,A_r)&\le
\sum_{(i,j):\,i\ne j}\Tr A_i^{1/2} A_j^{1/2}
=
\sum_{(i,j):\,i\ne j} \sqrt{p_i p_j} \;\;Q_{\min}(\sigma_i,\sigma_j)\label{eq:pureboundT}\\
&=
\sum_{(i,j):\,i\ne j} \frac{1}{\sqrt{p_i p_j}} \;\;F(A_i,A_j)^2
\le
\sum_{(i,j):\,i\ne j} \frac{p_i+p_j}{\sqrt{p_i p_j}} \;\;P_e^*(A_i,A_j).
\label{eq:pureboundT2}
\end{align}
\end{theorem}
\begin{proof}
The first inequalities in \eqref{eq:mixedboundT} and \eqref{eq:pureboundT} are due to \eqref{square bound} and
Proposition \ref{prop:half}.
The second inequality in  \eqref{eq:mixedboundT} is obvious from
$\Tr A_i^{1/2} A_j^{1/2}\le\|A_i^{1/2} A_j^{1/2}\|_1=F(A_i,A_j)$. The first bound in
\eqref{eq:mixedboundT2} follows from lemma \ref{lemma:FT bounds}, while the second bound is due to \eqref{fid chernoff}.
The identities in \eqref{eq:pureboundT} and \eqref{eq:pureboundT2} are straightforward to verify, and
the inequality in \eqref{eq:pureboundT2} is again due to lemma \ref{lemma:FT bounds}.
\end{proof}

\begin{rem}
Since the inequalities used to prove \eqref{eq:mixedboundT2} can be saturated, the square roots in \eqref{eq:mixedboundT2}
cannot be removed from $P_e^*(A_i,A_j)^{1/2}$ and $Q_{\min}(\sigma_i,\sigma_j)^{1/2}$ in general.
\end{rem}

\begin{rem}
When all states are pure and the prior is uniform (i.e., $p_k=\Tr A_k=1/r\s\forall k$), we can use another argument.
By the inequality $\trace\rho\sigma\le 1-(||\rho-\sigma||_1/2)^2$
\cite{sandwich}, we get
$$
\trace\sigma_j\sigma_k \le 1-\norm{\sigma_j-\sigma_k}_1^2/4 \le 2-\norm{\sigma_j-\sigma_k}_1.
$$
Hence,
$$
\trace A_j^{1/2}A_k^{1/2} = \frac{1}{r}\trace \sigma_j \sigma_k \le \frac{1}{r}(2-\norm{\sigma_j-\sigma_k}_1),
$$
so that
\be
1-\Gamma^* \ge  \left(1+\frac{2}{r} \sum_{(j,k):\,j<k}(2-\norm{\sigma_j-\sigma_k}_1)\right)^{-1/2}.\label{eq:purebound2}
\ee
\end{rem}

Based on extensive numerical simulations, 
we conjecture that the latter bound also holds for mixed states and for non-uniform priors:
\begin{conjecture}
For any $A_i\ge0$ with $\sum_i \trace A_i=1$,
\be
\Gamma^* \le 1- \left(1+4(r-1)P_{e,2}^*\right)^{-1/2} \le 2(r-1)P_{e,2}^*.\label{eq:mixedbound}
\ee
\end{conjecture}
By Theorem \ref{th:Tyson}, this would imply the inequality of Conjecture \ref{claim:1} for weighted states:
$$
P_e^* \le 4(r-1)P_{e,2}^*.
$$

\begin{rem}
Note that for any $p_i,p_j$, $(p_ip_j)^{3/2}\le p_ip_j\le (p_i^2+p_j)^2/2$, from which it follows that the constants
in the bound
\begin{align*}
P_e^*\le
\sum_{(i,j):\,i\ne j} \frac{1}{\sqrt{p_i p_j}} \;\;F(A_i,A_j)^2,
\end{align*}
given in \eqref{eq:pureboundT2}, are better than in \eqref{HLS bound}, i.e.,
\eqref{eq:pureboundT2} gives a tighter upper bound on the optimal error than \eqref{HLS bound}.

To compare the bounds in \eqref{BK bound} and \eqref{eq:mixedboundT}, first choose
 all the $\sigma_j$ to be pure, i.e., $\sigma_j=\pr{\psi_j}$ for some unit vectors $\psi_j$. Then
$\Tr A_i^{1/2}A_j^{1/2}=\sqrt{p_ip_j}|\inner{\psi_i}{\psi_j}|^2$, while
$F(A_i,A_j)=\sqrt{p_ip_j}|\inner{\psi_i}{\psi_j}|$. Choosing thus the $\psi_j$ close to orthogonal, but not orthogonal, we see that the ratio
\begin{align*}
\frac{\sum_{(i,j):\,i\ne j}\Tr A_i^{1/2} A_j^{1/2}}{\sum_{(i,j):\,i\ne j}F(A_i,A_j)}
\end{align*}
can be arbitrarily small. By continuity, we can also add a small perturbation to obtain PSD operators $A_j$ of full support
with the same property.
In this sense, the upper bound $P_e^*\le \sum_{(i,j):\,i\ne j}\Tr A_i^{1/2} A_j^{1/2}$ in \eqref{eq:mixedboundT}
can be arbitrarily better
than the bound in \eqref{BK bound}, for any fixed $r$. On the other hand, there are configurations for which
the bound in  \eqref{BK bound} outperforms the one in \eqref{eq:mixedboundT}, due to the $1/2$ pre-factor in the former.
\end{rem}

\section{Binary  state discrimination: i.i.d.~vs.~averaged i.i.d}
\label{sec:averaged}

Consider the binary state discrimination problem where one of the hypotheses is i.i.d., i.e.,
for $n$ copies it is represented by $\rho^{\otimes n}$ for some state $\rho$,
while the other hypothesis is averaged i.i.d., i.e., for $n$ copies it is of the form
$\sum_{i=1}^r q_i\sigma_i^{\otimes n}$ for some states $\sigma_1,\ldots,\sigma_r$, and a
probability distribution $q_1,\ldots,q_r$. This represents a situation where we have a
further uncertainty about the identity of the true state when the second hypothesis is true.
Alternatively, this can be considered as a state discrimination problem with $r+1$
i.i.d.~hypotheses, where we only want to know whether one of the hypotheses is true or not.
If the state $\rho$ has prior probability $0<p<1$ then the optimal error probability for $n$
copies is
\begin{equation*}
P_e^*\bz p\rho^{\otimes n},(1-p)\sum\nolimits_i q_i\sigma_i^{\otimes n}\jz
=
\half\bz 1-\norm{p\rho^{\otimes n}-(1-p)\sum\nolimits_i q_i\sigma_i^{\otimes n}}_1\jz.
\end{equation*}
Convexity of the trace-norm implies that
\begin{align}
P_e^*\bz p\rho^{\otimes n},(1-p)\sum\nolimits_i q_i\sigma_i^{\otimes n}\jz
&\ge
\half\sum_{i=1}^r q_i\bz
1-\norm{p\rho^{\otimes n}-(1-p)\sigma_i^{\otimes n}}_1
\jz\nonumber\\
&=
\sum_{i=1}^r q_i P_e^*\bz p\rho^{\otimes n},(1-p)\sigma_i^{\otimes n}\jz,\label{concavity of error}
\end{align}
and hence,
\begin{align*}
\pli\bz \{p\rho^{\otimes n}\}_n,\left\{(1-p)\sum\nolimits_i q_i\sigma_i^{\otimes n}\right\}_n\jz
\ge
\max_{1\le i\le r}\pli\bz \{p\rho^{\otimes n}\}_n,\{(1-p)\sigma_i^{\otimes n}\}_n\jz
=
-\min_i C(\rho,\sigma_i).
\end{align*}
Based on analytical proofs for various special cases as well extensive numerical search, we conjecture that the following converse decoupling inequality is also true:
\begin{conjecture}\label{con:averaged asymptotics}
\begin{align*}
\pls\bz \{p\rho^{\otimes n}\}_n,\left\{(1-p)\sum\nolimits_i q_i\sigma_i^{\otimes n}\right\}_n\jz
\le
\max_{1\le i\le r}\pls\bz \{p\rho^{\otimes n}\}_n,\{(1-p)\sigma_i^{\otimes n}\}_n\jz
=
-\min_i C(\rho,\sigma_i).
\end{align*}
\end{conjecture}
This conjecture would immediately yield
\begin{conjecture}\label{con:averaged asymptotics2}
\begin{align*}
\lim_{n\to\infty}\frac{1}{n}\log P_e^*\bz p\rho^{\otimes n},(1-p)\sum\nolimits_i q_i\sigma_i^{\otimes n}\jz
=
-\min_i C(\rho,\sigma_i).
\end{align*}
\end{conjecture}
\medskip

Below we will prove Conjecture \ref{con:averaged asymptotics} in the case where
$\rho$ is a pure state, and prove a weaker version in the general case.
These will follow from the following single-shot decoupling bounds, which are the main results of this section:

\begin{thm}\label{thm:averaged upper bounds}
Let $A,B_1,\ldots,B_r\in\B(\hil)_+$. Then
\begin{align}\label{averaged single-shot bounds}
P_e^*\bz A,\sum\nolimits_j B_j\jz\le
\sum_j F(A,B_j)\le
\sum_j\sqrt{\Tr (A+B_j)}\sqrt{P_e^*\bz A,B_j\jz}.
\end{align}
If $A$ is rank one then we also have
\begin{align}\label{pure upper bounds}
P_e^*\bz A,\sum\nolimits_j B_j\jz
\le
\begin{cases}
\bz\sum\nolimits_j\Tr B_j\jz\sum\nolimits_j F\bz \frac{A}{\Tr A},\frac{B_j}{\Tr B_j}\jz^2
\le
\bz\sum\nolimits_j\Tr B_j\jz\sum\nolimits_j\frac{\Tr A+\Tr B_j}{(\Tr A)(\Tr B_j)}P_e^*(A,B_j)\\
\bz\sum\nolimits_j\sqrt{1+\frac{\Tr B_j}{\Tr A}}\sqrt{P_e^*\bz A,B_j\jz}\jz^2
\le
\bz\sum\nolimits_j\bz 1+\frac{\Tr B_j}{\Tr A}\jz\jz
\sum\nolimits_j P_e^*\bz A,B_j\jz.
\end{cases}
\end{align}
\end{thm}
\bigskip

Before proving Theorem \ref{thm:averaged upper bounds}, we first explore some of its implications.
We start with the following:

\begin{cor}\label{cor:averaged asymptotics}
For every $n\in\bN$, let $A_n,B_{1,n},\ldots,B_{r,n}\in\B(\hil_n)_+$, where $\hil_n$ is some finite-dimensional Hilbert space.
If $\limsup_n\Tr(A_n+\sum\nolimits_j B_{j,n})<+\infty$ then
\begin{align*}
\pls\bz \vecc{A},\sum\nolimits_j \vecc{B}_j\jz\le
\half\max_{1\le j\le r}\pls\bz\vecc{A}, \vecc{B}_j\jz.
\end{align*}
If $A_n$ is rank one for every large enough $n$
and $\limsup_n\Tr(A_n+\sum\nolimits_j B_{j,n})/\Tr A_n<+\infty$
then
\begin{align*}
\pls\bz \vecc{A},\sum\nolimits_j \vecc{B}_j\jz\le
\max_{1\le j\le r}\pls\bz \vecc{A},\vecc{B_j}\jz.
\end{align*}
\end{cor}

\begin{rem}
Note that $P_e^*\bz A,\sum\nolimits_j B_j\jz\ge P_e^*\bz A,B_j\jz$ for every $j$, and hence
\begin{align*}
\max_{1\le j\le r} P_e^*\bz A,B_j\jz\le P_e^*\bz A,\sum\nolimits_j B_j\jz.
\end{align*}
In the asymptotic setting this yields
\begin{align*}
\max_{1\le j\le r}\pli\bz \vecc{A},\vecc{B}_j\jz\le\pli\bz \vecc{A},\sum\nolimits_j \vecc{B}_j\jz,
\end{align*}
complementing the inequalities of Corollary \ref{cor:averaged asymptotics}.
\end{rem}
\medskip

Applying Corollary \ref{cor:averaged asymptotics} to the problem of i.i.d.~vs.~averaged i.i.d.~state discrimination, we finally get the following:
\begin{thm}\label{thm:averaged iid}
In the i.i.d.~vs.~averaged i.i.d.~case described at the beginning of the section,
\begin{align*}
-\min_i C(\rho,\sigma_i)
&\le
\liminf_{n\to\infty}\frac{1}{n}\log P_e^*\bz p\rho^{\otimes n},(1-p)\sum\nolimits_i q_i\sigma_i^{\otimes n}\jz\\
&\le
\limsup_{n\to\infty}\frac{1}{n}\log P_e^*\bz p\rho^{\otimes n},(1-p)\sum\nolimits_i q_i\sigma_i^{\otimes n}\jz\le
-\half\min_i C(\rho,\sigma_i).
\end{align*}
If $\rho$ is pure then we have
\begin{align*}
\lim_{n\to\infty}\frac{1}{n}\log P_e^*\bz p\rho^{\otimes n},(1-p)\sum\nolimits_i q_i\sigma_i^{\otimes n}\jz
=
-\min_i C(\rho,\sigma_i).
\end{align*}
\end{thm}
\bigskip

In realistic scenarios it is more natural to assume that the hypotheses are represented by sets of states with many elements
(composite hypothesis) rather than one single state (simple hypothesis). Here we briefly consider the simplest such scenario,
where we have two hypotheses, of which one is simple, represented by some PSD operator $A$, and the other one is composite,
represented by a finite set of PSD operators $\{B_1,\ldots,B_r\}$. For a given POVM $\{E,I-E\}$,
the worst-case error probability is given by $\Tr A(I-E)+\max_{1\le i\le r}\Tr B_i E$, and we define
\begin{align*}
P_e^*\bz A,\{B_i\}_{i=1}^r\jz:=\inf\left\{\Tr A(I-E)+\max_{1\le i\le r}\Tr B_i E:\,0\le E\le I\right\}.
\end{align*}
For every $i$ and every $E$, we have
\begin{align*}
\Tr A(I-E)+ \Tr B_i E&\le \Tr A(I-E)+\max_{1\le i\le r}\Tr B_i E
\le
\Tr A(I-E)+\Tr \sum_{i=1}^rB_i E,
\end{align*}
and taking the infimum in $E$ yields
\begin{align*}
\max_{1\le i\le r}P_e^*\bz A,B_i\jz\le P_e^*\bz A,\{B_i\}_{i=1}^r\jz\le P_e^*\bz A,\sum\nolimits_i B_i\jz.
\end{align*}

Corollary \ref{cor:averaged asymptotics} then immediately yields the following:

\begin{cor}
For every $n\in\bN$, let $A_n,B_{1,n},\ldots,B_{r,n}\in\B(\hil_n)_+$, where $\hil_n$ is some finite-dimensional Hilbert space, and let
\begin{align*}
\pli\bz \vecc{A},\{\vecc{B}_i\}_{i=1}^r\jz&:=\liminf_{n\to+\infty}\frac{1}{n}\log P_e^*\bz A,\{B_i\}_{i=1}^r\jz\\
\pls\bz \vecc{A},\{\vecc{B}_i\}_{i=1}^r\jz&:=\limsup_{n\to+\infty}\frac{1}{n}\log P_e^*\bz A,\{B_i\}_{i=1}^r\jz.
\end{align*}
If $\limsup_n\Tr(A_n+\sum\nolimits_j B_{j,n})<+\infty$ then
\begin{align*}
\max_{1\le i\le r}\pli\bz\vecc{A}, \vecc{B}_i\jz\le
\pli\bz \vecc{A},\{\vecc{B}_i\}_{i=1}^r\jz\le
\pls\bz \vecc{A},\{\vecc{B}_i\}_{i=1}^r\jz\le
\half\max_{1\le i\le r}\pls\bz\vecc{A}, \vecc{B}_i\jz.
\end{align*}
If $A_n$ is rank one for every large enough $n$
and $\limsup_n\Tr(A_n+\sum\nolimits_j B_{j,n})/\Tr A_n<+\infty$
then
\begin{align*}
\max_{1\le i\le r}\pli\bz\vecc{A}, \vecc{B}_i\jz\le
\pli\bz \vecc{A},\{\vecc{B}_i\}_{i=1}^r\jz\le
\pls\bz \vecc{A},\{\vecc{B}_i\}_{i=1}^r\jz\le
\max_{1\le i\le r}\pls\bz\vecc{A}, \vecc{B}_i\jz.
\end{align*}
\end{cor}

Taking now $A_n:=\rho^{\otimes n},\,B_{i,n}:=\sigma_i^{\otimes n}$, where $\rho,\sigma_1,\ldots,\sigma_r$ are 
density operators on some finite-dimensional Hilbert space, we get the following analogous statement to Theorem \ref{thm:averaged iid}:
\begin{thm}
Let $\rho,\sigma_1,\ldots,\sigma_r$ be density operators on some finite-dimensional Hilbert space. Then
\begin{align*}
-\min_i C(\rho,\sigma_i)
&\le
\liminf_{n\to\infty}\frac{1}{n}\log P_e^*\bz \rho^{\otimes n},\{\sigma_i^{\otimes n}\}_{i=1}^r\jz\\
&\le
\limsup_{n\to\infty}\frac{1}{n}\log P_e^*\bz \rho^{\otimes n},\{\sigma_i^{\otimes n}\}_{i=1}^r\jz
\le
-\half\min_i C(\rho,\sigma_i).
\end{align*}
If $\rho$ is pure then we have
\begin{align*}
\lim_{n\to\infty}\frac{1}{n}\log P_e^*\bz \rho^{\otimes n},\{\sigma_i^{\otimes n}\}_{i=1}^r\jz
=
-\min_i C(\rho,\sigma_i).
\end{align*}
\end{thm}
\bigskip

Now we turn to the proof of  Theorem \ref{thm:averaged upper bounds}. For this we will need
the following subadditivity property of the fidelity:

\begin{lemma}\label{lemma:fidelity subadditivity}
Let $A,B_1,\ldots,B_r\in\B(\hil)_+$. Then
\begin{equation}\label{fidelity subadditivity}
F\bz A,\sum\nolimits_i B_i\jz\le
\sum\nolimits_i F(A,B_i).
\end{equation}
\end{lemma}
\begin{proof}
The function $X\mapsto\Tr\sqrt{X}$ is subadditive on PSD operators, i.e.,
if $X,Y\in\B(\hil)_+$ then $\Tr\sqrt{X+Y}\le\Tr\sqrt{X}+\Tr\sqrt{Y}$.
Indeed, assume first that $X,Y>0$. Then
\begin{align*}
\Tr\sqrt{X+Y}-\Tr\sqrt{X}&=
\int_{0}^1\frac{d}{dt}\Tr\sqrt{X+tY}\,dt
=
\int_{0}^1\Tr Y\half(X+tY)^{-1/2}\,dt\\
&\le
\Tr \sqrt{Y}\int_{0}^1\frac{t^{-1/2}}{2}\,dt
=
\Tr\sqrt{Y},
\end{align*}
where we used the identity $\frac{d}{dt}\Tr f(X+tY)=\Tr Yf'(X+tY)$, and
that the function $x\mapsto x^{-1/2}$ is operator monotone decreasing. The assertion for general PSD $X$ and $Y$ then follows by continuity.
Thus,
\begin{equation*}
F\bz A,\sum\nolimits_i B_i\jz=\Tr\sqrt{\sum\nolimits_i A^{1/2}B_iA^{1/2}}\le
\sum\nolimits_i \Tr\sqrt{A^{1/2}B_iA^{1/2}}=
\sum\nolimits_i F(A,B_i).\qedhere
\end{equation*}
\end{proof}
\medskip

After this preparation, we are ready to prove Theorem \ref{thm:averaged upper bounds}.
\smallskip

\noindent\textit{Proof of Theorem \ref{thm:averaged upper bounds}:}
\begin{align*}
P_e^*\bz A,\sum\nolimits_j B_j\jz
&\le
F\bz A,\sum\nolimits_j B_j\jz
\le
\sum_j  F\bz A,B_j\jz
\le
\sum_j \sqrt{\Tr (A+B_j)}\sqrt{P_e^*(A,B_j)},
\end{align*}
where we used Lemma \ref{lemma:FT bounds} in the first inequality,
the second inequality is due to Lemma \ref{lemma:fidelity subadditivity},
and the third inequality is again due to Lemma \ref{lemma:FT bounds}.
This proves \eqref{averaged single-shot bounds}.

Assume now that $A$ is rank one. Then
\begin{align*}
P_e^*\bz A,\sum\nolimits_j B_j\jz
&\le
\frac{1}{\Tr A}F\bz A,\sum\nolimits_j B_j\jz^2
\le
\frac{1}{\Tr A}\bz\sum\nolimits_j F(A,B_j)\jz^2\\
&=
\frac{1}{\Tr A}\bz\sum\nolimits_j (\Tr A)^{\half}(\Tr B_j)^{\half}F\bz\frac{A}{\Tr A},\frac{B_j}{\Tr B_j}\jz\jz^2\\
&\le
\frac{1}{\Tr A}\bz \sum\nolimits_j (\Tr A)(\Tr B_j)\jz\bz \sum\nolimits_j F\bz\frac{A}{\Tr A},\frac{B_j}{\Tr B_j}\jz^2\jz\\
&=
\bz\sum\nolimits_j \Tr B_j\jz\sum\nolimits_j\frac{1}{(\Tr A)(\Tr B_j)}F(A,B_j)^2\\
&\le
\bz\sum\nolimits_j \Tr B_j\jz\sum\nolimits_j\frac{\Tr A+\Tr B_j}{(\Tr A)(\Tr B_j)}P_e^*(A,B_j),
\end{align*}
where the first inequality is due to Lemma \ref{lem:FTpure},
the second inequality is due to Lemma \ref{lemma:fidelity subadditivity},
in the third inequality we used the Cauchy-Schwarz inequality, and the last inequality
follows from Lemma \ref{lemma:FT bounds}. This proves
the first bound in \eqref{pure upper bounds}.
Alternatively, we may proceed as
\begin{align*}
P_e^*\bz A,\sum\nolimits_j B_j\jz
&\le
\frac{1}{\Tr A}F\bz A,\sum\nolimits_j B_j\jz^2
\le
\frac{1}{\Tr A}\bz\sum\nolimits_j F(A,B_j\jz^2\\
&\le
\frac{1}{\Tr A}\bz\sum\nolimits_j \sqrt{\Tr (A+B_j)}\sqrt{P_e^*(A,B_j)}\jz^2\\
&\le
\frac{1}{\Tr A}\bz\sum\nolimits_j\Tr(A+B_j)\jz\sum\nolimits_jP_e^*(A,B_j),
\end{align*}
where the third inequality is due to Lemma \ref{lem:FTpure}, and
in the last line we used the Cauchy-Schwarz inequality. This proves the second bound  in \eqref{pure upper bounds}.
\qed
\bigskip

We close this section with some discussion of the above results.

Let $\rho,\sigma_1,\ldots,\sigma_r$ be states and $q_1,\ldots,q_r$ be a probability distribution.
Then we have
\begin{align*}
\sum_i q_i F(\rho,\sigma_i)\le F\bz\rho,\sum\nolimits_i q_i\sigma_i\jz\le
\sum_i \sqrt{q_i} F(\rho,\sigma_i),
\end{align*}
where the first inequality is a special case of the joint concavity of the fidelity
\cite[Theorem 9.7]{NC}, and the second inequality is due to Lemma \ref{lemma:fidelity subadditivity}
with the choice $A=\rho$ and $B_i=q_i\sigma_i$.
Hence, Lemma \ref{lemma:fidelity subadditivity} yields a complement to the concavity inequality
$\sum_i q_i F(\rho,\sigma_i)\le F(\rho,\sum_i q_i\sigma_i)$. It is natural to ask whether the joint concavity inequality 
$\sum_i q_i F(\rho_i,\sigma_i)\le F(\sum_iq_i\rho_i,\sum_i q_i\sigma_i)$
can be complemented in the same way, but it is easy to see that the answer is no. Indeed,
let $\rho_1=\sigma_2=\pr{x}$ and $\rho_2=\sigma_1=\pr{y}$ with $x,y$ being orthogonal unit vectors in $\bC^2$, and let $q_1=q_2=1/2$. 
Then $\sum_i q_i\rho_i=\half I=\sum_iq_i\sigma_i$, and hence $F(\sum_i q_i\rho_i,\sum_i q_i\sigma_i)=1$, while
$F(\rho_1,\sigma_1)=F(\rho_2,\sigma_2)=0$, and hence no inequality of the form
$F(\sum_iq_i\rho_i,\sum_i q_i\sigma_i)\le c\sum_i F(\rho_i,\sigma_i)$ can hold with some $c>0$.

One can ask the same questions about the quantity $P_e^*(.\, ,.)=\Tr\glb(.\, ,.)$, which has very similar properties to the fidelity. 
Indeed, convexity of the trace-norm yields joint concavity of this quantity, i.e.,
$\Tr\glb(\sum_i q_i\rho_i,\sum_i q_i\sigma_i)=
\half\bz 1-\norm{\sum_i q_i\rho_i-\sum_i q_i\sigma_i}_1\jz
\ge
\sum_i q_i\half\bz 1-\norm{\rho_i-\sigma_i}_1\jz
=
\sum_i q_i\Tr\glb(\rho_i,\sigma_i)$,
and the same example as above shows that this inequality cannot be complemented in general. On the other hand, 
one may hope that the weaker concavity inequality, where the first argument is a fixed $\rho$, can be complemented the 
same way as for the fidelity, i.e., that there exists a constant $c>0$, depending at most on $r$, such that
\begin{equation*}
\Tr\glb\bz\rho,\sum\nolimits_i q_i\sigma_i\jz=
\half\bz 1-\norm{\rho-\sum\nolimits_i q_i\sigma_i}_1\jz
\le
\frac{c}{2}\sum_i \bz 1-\norm{\rho-\sigma_i}_1\jz
=
c\sum_i \Tr\glb(\rho,\sigma_i).
\end{equation*}
More generally, one could ask whether an analogy of the subadditivity inequality \eqref{fidelity subadditivity} holds for 
$\Tr\glb(.\, ,.)$, i.e., if there exists a $c>0$, depending at most on $r$, such that
\begin{align}\label{eq:wrong1}
P_e^*\bz A,\sum\nolimits_j B_j\jz=\Tr\glb\bz A,\sum\nolimits_j B_j\jz\le c\sum_j\Tr\glb(A,B_j)=c\sum_j P_e^*(A,B_j)
\end{align}
holds for any PSD $A,B_1,\ldots,B_r$, where $c>0$ depends only on $r$.
This would give an improvement over Theorem \ref{thm:averaged upper bounds}, and prove Conjecture \ref{con:averaged asymptotics}.
Note that \eqref{eq:wrong1} is true when $A$ is of rank one, according to Theorem \ref{thm:averaged upper bounds}, and also when 
all the operators are commuting, as we show in Appendix \ref{sec:classical}.
However, as it turns out, no such $c$ exists in the general case.

Counterexamples are as follows: for $r=2$,
take $A=\ep|\psi_1\rangle\langle\psi_1|$, $B_1=|\psi_2\rangle\langle\psi_2|$ and $B_2=|\psi_3\rangle\langle\psi_3|$
with $\ep$ small and $\psi_2$ and $\psi_3$ very close and almost orthogonal to $\psi_1$.
For example, consider
$$
\psi_1 = \twovec{1}{0},\ds\ds
\psi_2 = \twovec{\sin\alpha}{\cos\alpha}, \ds\ds
\psi_3 = \twovec{-\sin\alpha}{\cos\alpha},
$$
with $\sin^2\alpha=\ep/2$.
Then $B_1+B_2=\diag(\ep,2-\ep)$ and
$\trace\GLB(A,B_1+B_2)=\half(\ep+2-|2-\ep|)=\ep$.
However, one can check that
$\trace\GLB(A,B_1)=\trace\GLB(A,B_2) \approx \ep^2/2$ for very small $\ep$.
Thus, the LHS of (\ref{eq:wrong1}) is linear in $\ep$, whereas its RHS is quadratic,
meaning that the RHS can be arbitrarily smaller than the LHS in the sense that
the RHS/LHS ratio can be arbitrarily small.

One might get the impression that this failure is due to the fact that $A_1$ has very small trace.
Thus one could try to amend inequality (\ref{eq:wrong1}) by dividing the RHS by that trace (making both sides linear in $\ep$):
\be
\trace\glb\bz A,\sum\nolimits_j B_j\jz \le
\frac{1}{\trace A}\; \sum_{j} \trace\glb(A,B_j).\label{eq:wrong2}
\ee
This is a sensible amendment as it resonates with the appearance of the factor
$\frac{1}{\Tr A}$ in \eqref{F^2 bound}
in our treatment of the pure state case, and furthermore, initial numerical simulations seemed to bolster the claim.
However, this inequality is false too.
We can use the direct sum trick based on Lemma \ref{lub direct sum},
and replace $A$ by $A\oplus(1-\Tr A)\pr{x}$
and $B_i$ by $B_i\oplus 0$ in the counterexample of the previous paragraph, where $x$ is a unit vector in some auxiliary Hilbert space.
This does not change the $\trace\GLB$ terms
but changes $\trace A$ to 1, thereby eliminating its supposedly compensating effect. Thus inequality (\ref{eq:wrong2})
is violated to arbitrarily high extent. Moreover, the same argument excludes the possibility to fix
inequality (\ref{eq:wrong2}) by replacing $1/\Tr A$ with
$f(\Tr A)$ where $f:\,\bR_+\to\bR_+$ is such that $\lim_{x\to 0^+}f(x)=+\infty$.

The problems presented by the above example could be eliminated if we allowed the cross-term
$\Tr\glb(B_1,B_2)$ to appear (with some positive constant) on the RHS of
\eqref{eq:wrong1}, since the term $\trace\GLB(B_1,B_2)$ is
close to 1 and swamps the distinction between $\ep$ and $\ep^2$.
Although such a bound is too weak for proving Conjecture \ref{con:averaged asymptotics}, it would be just the right tool to prove
Conjecture \ref{claim:1}, as we will see in the next section.

\section{Dichotomic discrimination}
\label{sec:dichotomic}

Consider the generalized state discrimination problem with hypotheses
$A_1,\ldots,A_r$.
In this section we show an intermediate step towards proving Conjecture \ref{claim:1} in the form of a partial decoupling bound. Namely, we prove
(in Theorem \ref{th:dich}) that
the multiple state discrimination error is bounded from above by the \ki{dichotomic error}, which is the sum of the error probabilities 
of discriminating one $A_i$ from the rest of the hypotheses.
Using then the bounds obtained in Section \ref{sec:averaged}, we get full decoupling bounds (Theorem \ref{thm:dichotomic decoupling}).

Define the complementary operators as
$\comp{A}_i:=\sum_{j\neq i}A_j=A_0-A_i$, where $A_0:=\sum_i A_i$.
If we only want to decide whether
the true hypothesis is $A_i$ or not, i.e., we want to discriminate between $A_i$ and
$\comp{A}_i$, then the corresponding optimal error is given by
\bea
P^*_{e,dich,i} := P_e^*(A_i,\comp{A}_i) = \trace\glb(A_i,\comp{A}_i)
=
\half(\Tr A_0-\norm{A_i-\comp{A}_i}_1).
\eea
We will call this a \textit{dichotomic} discrimination, and
$P^*_{e,dich,i}$ the $i$-th optimal \ki{dichotomic error}.
Let us define $P_{e,dich}^*$ as the sum of the $r$ optimal dichotomic errors corresponding to each of the $A_i$:
\be
P^*_{e,dich}(A_1,\ldots,A_r) := P^*_{e,dich}:= \sum_{i=1}^r P^*_{e,dich,i}= \sum_{i=1}^r P_e^*(A_i,\comp{A}_i).
\ee

We show in Theorem \ref{th:dich} that the optimal multi-hypothesis discrimination error $P_e^*$ is well-approximated by
 $P^*_{e,dich}$; more precisely,
\begin{equation}\label{dichotomic bounds}
\half P^*_{e,dich} \le  P^*_e \le P^*_{e,dich}.
\end{equation}
In particular, these bounds, together with the fact that $P^*_{e,dich}$ is a number between 0 and $\Tr A_0$, show that there exists a POVM
$\{E_k\}_{k=1}^r$ for which $P_e\bz\{E_k\}_{k=1}^r\jz=P^*_{e,dich}$.
Therefore, we can rightly call $P^*_{e,dich}$ the \textit{dichotomic error}. Moreover, these inequalities show that $P_e^*$ and $P^*_{e,dich}$ have the same exponential
behavior in the limit of many i.i.d.~copies of the hypotheses.

\medskip

We need some preparation to prove the bounds in \eqref{dichotomic bounds}.
First, we give a number of useful expressions for $P^*_{e,dich}$.
Since $A_i-\comp{A}_i=2 A_i-A_0$,
we have $||A_i-\comp{A}_i||_1 = 2\trace(2 A_i-A_0)_+ + \Tr A_0-2\Tr A_i$. Then
\begin{align}
P^*_{e,dich} &=
 \half \sum_{i=1}^r (\Tr A_0-||A_i-\comp{A}_i||_1) \nonumber\\
&= \frac{1}{2} \sum_{i=1}^r 2\Tr A_i-2\trace(2 A_i-A_0)_+ \nonumber\\
&= \Tr A_0-\sum_{i=1}^r \trace(2 A_i-A_0)_+ \nonumber\\
&= \Tr A_0-\sum_{i=1}^r \trace(A_i-\comp{A}_i)_+ .\label{dich expressions}
\end{align}
These expressions show that the quantity $P^*_{e,dich}$ is a number between 0 and $\Tr A_0$.

Next, we prove Lemma \ref{lem:proj} below, which we will use for the proof of the upper bound in \eqref{dichotomic bounds}.
Note that the map $f:\,X\mapsto X^*X$ is operator convex on $\B(\hil)$, as it was pointed out in \cite[Lemma 5]{OH}.
Indeed, for any $X_1,X_2\in\B(\hil)$ and any $t\in[0,1]$, we have
\begin{align*}
tf(X_1)+(1-t)f(X_2)-f(tX_1+(1-t)X_2)=
t(1-t)(X_1-X_2)^*(X_1-X_2)\ge 0.
\end{align*}
In particular, for $X_1,\ldots,X_r\in\B(\hil)$, we have
$\bz\sum_i X_i\jz^*\bz\sum_i X_i\jz\le r\sum_i X_i^*X_i$, and operator monotony of the square root yields
\begin{align}\label{convexity consequence}
\left|\sum\nolimits_i X_i\right|\le \sqrt{r} \bz\sum\nolimits_i |X_i|^2\jz^{1/2}.
\end{align}

\begin{lemma}\label{lem:proj}
Let $\{P_i\}_{i=1}^r$ be a set of $r$ projectors. Define $P_0=\sum_{i=1}^r P_i$.
Then
$$
0\le\sum_i (2P_i-P_0)_+ \le\id,
$$
i.e.\ the set of operators $\{(2P_i-P_0)_+\}_{i=1}^r$ forms an (incomplete) POVM.
\end{lemma}
\begin{proof}
Let $X_i:=|2P_i-P_0|$,
with $P_i$ and $P_0$ as defined in the statement of the lemma.
By \eqref{convexity consequence},
$$
\sum\nolimits_i |2P_i-P_0| \le \sqrt{r} \bz\sum\nolimits_i (2P_i-P_0)^2\jz^{1/2}.
$$
Considering the facts that the $P_i$ are projectors, i.e.\ $P_i^2=P_i$, and that $P_0$ is equal to their sum,
the expression $\sum_i (2P_i-P_0)^2$ simplifies to
$$
\sum_i (2P_i-P_0)^2 = \sum_i(4P_i+P_0^2-2P_iP_0-2P_0P_i) = 4P_0+rP_0^2-4P_0^2 = 4P_0+(r-4)P_0^2.
$$
Now note the following:
\beas
r(4P_0+(r-4)P_0^2) &\le& r(4P_0+(r-4)P_0^2) + 4(I-P_0)^2 \\
&=& 4rP_0+(r^2-4r)P_0^2+4I-8P_0+4P_0^2 \\
&=& 4I+4(r-2)P_0+(r-2)^2P_0^2 \\
&=& (2I+(r-2)P_0)^2.
\eeas
Thus, we get
$$
\sum_i |2P_i-P_0| \le 2I+(r-2)P_0.
$$
To rewrite this in terms of the positive parts, we use the relation $|X| = 2X_+ - X$. This gives
$$
\sum_i |2P_i-P_0|
= 2\sum_i (2P_i-P_0)_+ - \sum_i (2P_i-P_0)
= 2\sum_i (2P_i-P_0)_+ - (2-r)P_0.
$$
Hence, we finally obtain
$$
\sum_i (2P_i-P_0)_+ = \half\left(\sum_i |2P_i-P_0| + (2-r)P_0\right) \le I,
$$
as we set out to prove.
\end{proof}
\medskip

Now we are ready to prove \eqref{dichotomic bounds}.

\begin{theorem}\label{th:dich}
For any $A_1,\ldots,A_r\in\B(\hil)_+$,
\begin{equation}\label{dich bounds2}
P_{e,2}^* \le \half P^*_{e,dich} \le  P^*_e \le P^*_{e,dich}.
\end{equation}
\end{theorem}
\begin{proof}
The first inequality follows by a straightforward computation:
\begin{align*}
P_{e,2}^*
&=
\frac{1}{r-1} \sum_{(k,l):\,k< l} P_e^*(A_k,A_l)
=
\frac{1}{2(r-1)} \sum_{(k,l):\,k\neq l}  P_e^*(A_k,A_l)\\
&=
\half \sum_k \frac{1}{r-1}\sum_{l:\, l\neq k}  P_e^*(A_k,A_l)
\le
\half \sum_k \frac{1}{r-1}\sum_{l: l\neq k}  P_e^*(A_k,\bar A_k)
=
\half \sum_k P_e^*(A_k,\bar A_k)
=
\half P^*_{e,dich}.
\end{align*}
The inequality is due to the fact that $A_l\le \bar A_k$ for $l\neq k$, and hence
$P_e^*(A_k,A_l)\le P_e^*(A_k,\bar A_k) $.

Next we prove the second inequality.
Let $\{E_i\}_{i=1}^r$ be the optimal POVM for $P_e^*$. Clearly,
$$
\trace(2 A_i-A_0)_+ \ge \trace(2 A_i-A_0)E_i = 2\trace A_i E_i-\trace A_0 E_i.
$$
Summing over $i$ yields
\begin{align*}
\sum_{i=1}^r \trace(2 A_i-A_0)_+
\ge  2\sum_{i=1}^r \trace A_iE_i- \trace A_0 \sum_{i=1}^r E_i
\ge 2\sum_{i=1}^r \trace A_iE_i-\Tr A_0
= 2P_s^*-\Tr A_0.
\end{align*}
Hence, by \eqref{dich expressions},
$$
P^*_{e,dich}=\Tr A_0-\sum_{i=1}^r \trace(2 A_i-A_0)_+ \le \Tr A_0-(2P_s^*-\Tr A_0) = 2P_e^*.
$$

We will use Lemma \ref{lem:proj} to prove the last inequality in \eqref{dich bounds2}.
The trace of the positive part $X_+$ of a Hermitian operator $X$ can be expressed as
$\trace XP$ with $P$ the projector on the support of $X_+$.
In particular, if $P_i$ is the projector on the support of $(A_i-\comp{A}_i)_+$, we have
$$
\sum_{i=1}^r\trace(A_i-\comp{A}_i)_+ = \sum_i\trace(A_i-\comp{A}_i)P_i.
$$
Defining $P_0:=\sum_i P_i$,
the summation on the right-hand side can be rewritten in the following way:
\begin{align*}
\sum_i\trace(A_i-\comp{A}_i)P_i
&=
\sum_i\trace(2A_i-A_0)P_i
= 2\sum_i\trace A_iP_i - \trace A_0 P_0
=\sum_i\trace(2P_i-P_0)A_i\\
&\le\sum_i \trace(2P_i-P_0)_+ A_i
\le
\max_{\{E_i\}\text{ POVM}}\sum_i\Tr E_i A_i=P_s^*,
\end{align*}
where the last inequality follows from the fact that
the set of operators $\{(2P_i-P_0)_+\}_{i=1}^r$ forms an (incomplete) POVM
by Lemma \ref{lem:proj}.
Hence
\begin{equation*}
P^*_{e,dich}=\Tr A_0-\sum_i\trace(A_i-\comp{A}_i)P_i \ge\Tr A_0-P_s^*= P_e^*.
\qedhere
\end{equation*}
\end{proof}
\smallskip

\begin{rem}
Validity of the last inequality in \eqref{dich bounds2}
in the classical case is a simple consequence of the fact that in a
list of positive numbers only the largest one can be bigger than half their sum. Hence, for diagonal states
$$
\sum_i (2A_i-A_0)_+ = \lub\left( \left\{(2A_i-A_0)_+\right\}\right) = (2\lub(\{A_i\})-A_0)_+
\le \frac{r\lub(\{A_i\})-A_0}{r-1}.
$$
Taking the trace then yields
$$
\Tr A_0-P^*_{e,dich} = \sum_i \trace(2A_i-A_0)_+ \le \frac{r\trace\lub(\{A_i\})-\Tr A_0}{r-1}
= \frac{rP_s^*-\Tr A_0}{r-1} = \Tr A_0-\frac{r}{r-1}P_e^*,
$$
which is slightly stronger than what we needed to prove.
\end{rem}

Note that the proof presented above for the first two inequalities in \eqref{dich bounds2} 
gives an alternative proof of the inequality $P_{e,2}^* \le  P^*_e$ from Theorem \ref{th:41}.
Moreover, we have obtained a strengthening of this inequality, by including $\half P^*_{e,dich}$ in between $P_{e,2}^*$ and $P^*_e$.

Theorem \ref{th:dich} shows that the pairwise error does not exceed one half of the dichotomic error, and we conjecture 
that it can not be less than the dichotomic error up to another constant factor (depending only on the
number of hypotheses). More precisely, we have the following:
\begin{conjecture}\label{claim:2}
There exists a constant $c$, at most depending on the number of hypotheses $r$, such that for all $A_1,\ldots,A_r\in\B(\hil)_+$,
$$
P^*_{e,dich}(A_1,\ldots,A_r) \le c P^*_{e,2}(A_1,\ldots,A_r).
$$
Explicitly,
\begin{align}
P^*_{e,dich}=\sum_{i=1}^r \trace\glb(A_i,\bar A_i)
&\le c(r)\frac{1}{r-1} \sum_{(i,j):\,i\ne j}\trace\glb(A_i,A_j)\nonumber\\
&=
\tilde c(r)\sum_{(i,j):\,i\ne j}P_e^*(A_i,A_j).
\label{eq:claim2}
\end{align}
\end{conjecture}
Numerical simulations suggest that $c(r)=4(r-1)$ is best possible.
Clearly, validity of this conjecture would prove validity of Conjecture \ref{claim:1}.
We can prove this conjecture for pure states, and for commuting states (see Appendix \ref{sec:classical}), 
whereas for mixed states we are able to prove a weaker inequality:

\begin{thm}\label{thm:dichotomic decoupling}
Let $A_1,\ldots,A_r\in\B(\hil)_+$ and $p_i:=\Tr A_i$. Then
\begin{align}\label{mixed dich upper}
P^*_{e}(A_1,\ldots,A_r)\le P^*_{e,dich}(A_1,\ldots,A_r)\le
\sum_{(i,j):\,i\ne j}F(A_i,A_j)
\le
\sum_{(i,j):\,i\ne j}\sqrt{p_i+p_j}\sqrt{P_e^*(A_i,A_j)}.
\end{align}
If $A_i$ is rank one for all $i$ then
\begin{align}\label{pure dich upper}
P^*_{e}(A_1,\ldots,A_r)&\le P^*_{e,dich}(A_1,\ldots,A_r)\nonumber\\
&\le
\begin{cases}
(\Tr A_0)\sum_{(i,j):\,i\ne j}\frac{1}{p_ip_j}F(A_i,A_j)^2\le
(\Tr A_0)\sum_{(i,j):\,i\ne j}\frac{p_i+p_j}{p_ip_j}P_e^*(A_i,A_j),\\
\frac{\Tr A_0}{\min_i\Tr A_i}\sum_{(i,j):\,i\ne j}P_e^*\bz A_i,A_j\jz.
\end{cases}
\end{align}
\end{thm}
\begin{proof}
The inequality $P^*_{e}(A_1,\ldots,A_r)\le P^*_{e,dich}(A_1,\ldots,A_r)$ is due to Theorem
\ref{th:dich}, and the rest
is immediate from the definition of $P^*_{e,dich}$ and Theorem \ref{thm:averaged upper bounds}.
\end{proof}

\begin{rem}
Note that the bound $P^*_{e}(A_1,\ldots,A_r)\le\sum_{(i,j):\,i\ne j}F(A_i,A_j)$ in \eqref{mixed dich upper}
is the same as in \eqref{eq:mixedboundT}, but weaker than the bound in \eqref{BK bound}, due to the $1/2$ prefactor in the latter.

To compare the bounds in \eqref{pure dich upper} to the other bounds obtained previously, we consider the most relevant case where
$\Tr A_0=p_1+\ldots +p_r=1$. Then \eqref{pure dich upper} tells that
$P^*_{e}(A_1,\ldots,A_r)\le\sum_{(i,j):\,i\ne j}\frac{1}{p_ip_j}F(A_i,A_j)^2$.  Since
$\frac{1}{p_ip_j}\le \half\frac{p_i^2+p_j^2}{p_i^2p_j^2}$, this bound is better than the one in
\eqref{HLS bound}, and the two coincide if and only if $p_1=\ldots=p_r$. On the other hand,
$\frac{1}{p_ip_j}>\frac{1}{\sqrt{p_ip_j}}$ (we assume that all $p_i>0$), and hence
the fidelity bound in \eqref{pure dich upper} is strictly worse than the one in
\eqref{eq:pureboundT2}.
\end{rem}
\bigskip

To close this section, we formulate two further conjectures that would imply Conjecture \ref{claim:2}.
We have seen in the previous section that no bound of the form
$\trace\glb(A_1,\overline{A_1}) \le c \sum_{l=2}^r\trace\glb(A_1,A_l)$ may hold in general, 
but amending the RHS with cross terms, i.e., error probabilities between
$A_k,A_l,\,k,l\ne 1$ may yield a valid upper bound. Although such a bound would not have been 
useful for the purposes of Section \ref{sec:averaged}, it would be sufficient for
Conjecture \ref{claim:2}, and numerical simulations suggest that it is indeed true.
Hence, we have the following
\begin{conjecture}\label{claim:3}
There exist constants $c_1$ and $c_2$, at most depending on the number of hypotheses $r$, such that for all $A_i\ge0$,
\be
\trace\glb(A_1,\bar A_1) \le c_1 \sum_{l=2}^r\trace\glb(A_1,A_l) + c_2 \sum_{k,l=2:k\neq l}^r\trace\glb(A_k,A_l).\label{eq:claim3}
\ee
\end{conjecture}
An equivalent conjecture in terms of POVM elements (using the primal SDP characterization of error probabilities) is:
\begin{conjecture}\label{claim:3b}
There exist constants $c_1$ and $c_2$, at most depending on the number of hypotheses $r$, such that for any
$0\le F_i\le\id$ (for $2\le i\le r$) and $0\le G_{j,k}\le\id$ (for $2\le j<k\le r$)
there exists an $E$ in the intersection of operator intervals
\be
\left.
\begin{array}{r}
0\\
\id-c_1\sum_{i=2}^r(\id-F_i)
\end{array}
\right\}
\le E\le
\left\{
\begin{array}{l}
\id \\
c_1 F_j + c_2\left(\sum_{k=j+1}^r G_{j,k} + \sum_{k=2}^{j-1}(\id-G_{k,j})\right),\quad j=2,\ldots,r.
\end{array}
\right.
\label{eq:claim3b}
\ee
\end{conjecture}
Note, however, that operator intervals behave very differently than ordinary intervals of real numbers and are not very well
understood. See for example the papers by Ando on this subject (e.g.\ \cite{ando93}).

\medskip

\noindent\textit{Proof of equivalence of Claims \ref{claim:3} and \ref{claim:3b}.}
The correspondence between the two claims is based on the following equivalent characterizations of the error probabilities:
\beas
\trace\glb(A_1,\overline{A_1}) &=& \min_E \trace((\id-E)A_1 + E\sum_{j=2}^r A_j) \\
\trace\glb(A_1,A_i) &=& \min_{F_i} \trace((\id-F_i)A_1 +F_iA_i) \\
\trace\glb(A_j,A_k) &=& \min_{G_{j,k}} \trace(G_{j,k}A_j +(\id-G_{j,k})A_k),
\eeas
where $E, F_i, G_{j,k}$ are POVM elements and satisfy $0\le E,F_i,G_{j,k}\le\id$.
Hence (\ref{eq:claim3}) holds if and only if
\beas
0 &\le& c_1 \sum_{i=2}^r \min_{F_i} \trace((\id-F_i)A_1 +F_iA_i) \\
&& + c_2 \sum_{j=2}^r \sum_{k=2: k\neq j}^r \min_{G_{j,k}} \trace(G_{j,k}A_j +(\id-G_{j,k})A_k) \\
&& - \min_E \trace((\id-E)A_1 + E\sum_{j=2}^r A_j) \\
&=&
\min_{F_i} c_1 \sum_{i=2}^r  \trace((\id-F_i)A_1 +F_iA_i) \\
&& + \min_{G_{j,k}} c_2 \sum_{j=2}^r \sum_{k=2: k\neq j}^r  \trace(G_{j,k}A_j +(\id-G_{j,k})A_k) \\
&& + \max_E -\trace((\id-E)A_1 + E\sum_{j=2}^r A_j) \\
&=& \min_{F_i} \min_{G_{j,k}} \max_E
\trace A_1\left(E-\id+c_1\sum_{i=2}^r (\id-F_i)\right) \\
&& + \sum_{j=2}^r \trace A_j\left(-E+c_1 F_j+c_2\left(\sum_{k=j+1}^r G_{j,k} + \sum_{k=2}^{j-1}(\id-G_{k,j})\right)\right)
\eeas
holds for all $A_i\ge0$.
This quantification can be rephrased as the requirement that the minimization of the RHS over all $A_i\ge0$ is non-negative.
By von Neumann's minimax theorem, the order between this minimization and the minimization over $E$  can be interchanged:
\beas
0&\le&\min_{F_i} \min_{G_{j,k}} \max_E
\min_{A_1\ge0} \trace A_1\left(E-\id+c_1\sum_{i=2}^r (\id-F_i)\right) \\
&& + \sum_{j=2}^r \min_{A_j\ge0}\trace A_j\left(-E+c_1F_j+c_2\left(\sum_{k=j+1}^r G_{j,k} + \sum_{k=2}^{j-1}(\id-G_{k,j})\right)\right).
\eeas
The minimizations over $F_i$ and $G_{j,k}$ and the maximization over $E$ can now be replaced by quantifications:
for all POVM elements $F_i$ and $G_{j,k}$ there should exist a POVM element $E$ such that
\beas
0&\le&
\min_{A_1\ge0} \trace A_1\left(E-\id+c_1\sum_{i=2}^r (\id-F_i)\right) \\
&& + \sum_{j=2}^r \min_{A_j\ge0}\trace A_j\left(-E+c_1 F_j+c_2\left(\sum_{k=j+1}^r G_{j,k} + \sum_{k=2}^{j-1}(\id-G_{k,j})\right)\right).
\eeas
Since $\trace AB\ge0$ for all $A\ge0$ if and only if $B\ge0$, this is so if and only if
$$
E-\id+c_1\sum_{i=2}^r (\id-F_i) \ge 0
$$
and, for all $j\ge2$,
$$-E+c_1F_j+c_2\left(\sum_{k=j+1}^r G_{j,k} + \sum_{k=2}^{j-1}(\id-G_{k,j})\right)\ge0.
$$
Combining this with the requirement $0\le E\le \id$ yields the inequalities of Claim \ref{claim:3b}.
\qed

\section{Nussbaum's mixed exponents approach}\label{sec:nussbaum}
In \cite{N13} Nussbaum presented a different approach towards splitting up the multi-hypothesis testing problem
into pairwise tests,
in which one pair of hypotheses is treated in a preferential way.
This leads to an upper bound on the total error probability in which different pairwise error probabilities
appear with different exponents. Here we generalize his approach and by combining it with our results
we improve his bounds on the total error probability.

First we need a lemma about POVM elements, the content of which is implicit in \cite{N13}:
\begin{lemma}\label{lem:POVMnb}
For any $E$ and $Q$  satisfying $0\le E,Q\le\id$,
\be
\half Q^{1/2}EQ^{1/2} \le \id-Q+E.\label{eq:nussbaum}
\ee
\end{lemma}
\begin{proof}
For any operator $X$, we have $0\le (X-2)^*(X-2)$, which can be rewritten as
$X^*X/2 \le (\id-X)^*(\id-X)+\id$.
In particular, let $E$ be positive definite and $Q$ positive semidefinite and let $X=E^{1/2}Q^{1/2}E^{-1/2}$.
Then we obtain
$$
\half E^{-1/2}Q^{1/2}EQ^{1/2}E^{-1/2} \le E^{-1/2}(\id-Q^{1/2})E(\id-Q^{1/2})E^{-1/2}+\id,
$$
which yields, after multiplying with $E^{1/2}$ on the left and on the right,
$$
\half Q^{1/2}EQ^{1/2} \le (\id-Q^{1/2})E(\id-Q^{1/2})+E.
$$
By continuity, this inequality also holds for positive semidefinite $E$.
If we now impose $E,Q\le\id$ then the RHS can be bounded above by a simplified expression:
\begin{equation*}
(\id-Q^{1/2})E(\id-Q^{1/2})+E \le (\id-Q^{1/2})^2+E \le \id-Q+E.
\qedhere
\end{equation*}
\end{proof}

Nussbaum's result relies on the following decomposition lemma, proven by him for the case of uniform priors and for $K=2$.
We provide the lemma in full generality, and with a somewhat shorter proof, but still based on Nussbaum's main idea to
decompose the POVM in a clever way into two parts.

\begin{lemma}\label{lem:nussbaum}
Let $A_1,\ldots,A_r\in\B(\hil)_+$. For all
$1\le K\le r$,
\begin{align}\label{Nussbaum bound}
P_e^*(A_1,\ldots,A_r) \le 2P_e^*(A_1,\ldots,A_K) + P_e^*(3A^{(K)},A_{K+1},\ldots,A_r),
\end{align}
where $A^{(K)}:=\sum_{i=1}^K A_i$.
\end{lemma}
\begin{proof}
Let $\cF=\{F_1,\ldots,F_K\}$ be the optimal POVM for discriminating between $A_1,\ldots,A_K$,
and let
$\cE^{(K)}=\{Q,E_{K+1},\ldots,E_r\}$ be the optimal POVM for
discriminating between $3A^{(K)}$, $A_{K+1},\ldots,A_r$. Define
$E_i:=Q^{1/2}F_iQ^{1/2}$ for $i=1,\ldots,K$. Then
$\cE=\{E_1,\ldots,E_r\}$ is a POVM.

In terms of the POVM $\cE$ we have, for $i=1,\ldots,K$,
$$
\trace A_i E_i = \trace A_i(Q-(Q-E_i))=\trace A_i Q - \trace A_i Q^{1/2}(\id-F_i) Q^{1/2}.
$$
The total error probability for this POVM (an upper bound on $P_e^*$) is given by
\beas
P_e\bz \cE\jz &=& \sum_{i=1}^K \trace A_i(\id-E_i)+ \sum_{i=K+1}^r\trace A_i(\id-E_i)  \\
&=&
\trace A^{(K)}(\id-Q) +\sum_{i=1}^K \trace A_i Q^{1/2}(\id-F_i) Q^{1/2} + \sum_{i=K+1}^r \trace A_i (\id-E_i).
\eeas
Using (\ref{eq:nussbaum}) of Lemma \ref{lem:POVMnb}, the second sum can be bounded above by
$$
2\left(\trace \sum_{i=1}^K A_i(\id-F_i) + \trace A^{(K)}(\id-Q)\right)
=
2P_e^*(A_1,\ldots,A_K)+ 2\trace A^{(K)}(\id-Q).
$$
Then
\begin{align*}
P_e^*\le P_e\bz \cE\jz  &\le
2P_e^*(A_1,\ldots,A_K)+3\trace A^{(K)}(\id-Q) + \sum_{i=K+1}^r \trace A_i (\id-E_i)\\
&=
2P_e^*(A_1,\ldots,A_K)+
P_e^*\bz 3A^{(K)},A_{K+1},\ldots,A_r\jz,
\end{align*}
proving \eqref{Nussbaum bound}.
\end{proof}

The above lemma yields immediately the following:
\begin{thm}
Let $A_1,\ldots,A_r\in\B(\hil)_+$. Then
\begin{align}\label{Nussbaum decoupling}
P_e^*(A_1,\ldots,A_r)\le
2^{r-2} P_e^*(A_1,A_2)+3\sum_{k=2}^{r-1} 2^{r-1-k} P_e^*\bz \sum_{i=1}^k A_i,A_{k+1}\jz.
\end{align}
\end{thm}
\begin{proof}
Applying Lemma \ref{lem:nussbaum} recursively, we get
\begin{align*}
P_e^*(A_1,\ldots,A_r)
&\le 2P_e^*(A_1,\ldots,A_{r-1}) + P_e^*(3A^{(r-1)},A_r) \\
&\le 4P_e^*(A_1,\ldots,A_{r-2}) + 2P_e^*(3A^{(r-2)},A_{r-1})+P_e^*(3A^{(r-1)},A_r) \\
&\le \ldots \\
&\le 2^{r-2} P_e^*(A_1,A_2)+\sum_{k=2}^{r-1} 2^{r-1-k} P_e^*(3A^{(k)},A_{k+1}).
\end{align*}
Note that $P_e^*(3A^{(k)},A_{k+1})\le 3P_e^*(A^{(k)},A_{k+1})$, and thus we obtain \eqref{Nussbaum decoupling}.
\end{proof}

\begin{rem}
Note that the upper bound in \eqref{Nussbaum decoupling} is similar to the bound $P_e^*\le P^*_{e,dich}$ in Theorem \ref{th:dich},
but the two are not directly comparable regarding their tightness.
\end{rem}

Combining now Theorem \ref{thm:averaged upper bounds} with the above theorem, we finally get the following 
decoupling bound in terms of the optimal pairwise error probabilities:
\begin{theorem}\label{th:nussbaum}
Let $A_1,\ldots,A_r\in\B(\hil)_+$ and
$\kappa:=3\max_{1\le i<j\le r}\sqrt{\Tr A_i+\Tr A_j}$. Then
\begin{align}\label{nussbaum mixed}
P_e^*(A_1,\ldots,A_r)
\le
2^{r-2} P_e^*(A_1,A_2)+\kappa\sum_{k=2}^{r-1} 2^{r-1-k} \sum_{l=1}^k \sqrt{P_e^*(A_l,A_{k+1})}.
\end{align}
If $A_k$ is of rank one for $k=3,\ldots,r$, then
\begin{equation}\label{nussbaum pure}
P_e^*(A_1,\ldots,A_r)\le
2^{r-2} P_e^*(A_1,A_2)+\kappa'\sum_{k=2}^{r-1} 2^{r-1-k} \sum_{l=1}^k  P_e^*(A_l,A_{k+1}),
\end{equation}
where $\kappa':=3\Tr A_0/\bz\min_{3\le i\le r}\Tr A_i\jz$.
\end{theorem}
\begin{proof}
Applying Theorem \ref{thm:averaged upper bounds} to each term in the summand in \eqref{Nussbaum decoupling}
yields the inequalities of the theorem.
\end{proof}

\begin{rem}
The constants in \eqref{nussbaum mixed} and \eqref{nussbaum pure} are in general worse than the ones in
Theorems \ref{th:half} and \ref{thm:dichotomic decoupling}. On the other hand, \eqref{nussbaum mixed}
outperforms all the previous bounds in the sense that for one pair of states, it contains the optimal binary error probability
instead of its square root. We will explore the consequences of this in the next section.
\end{rem}

\section{Asymptotics: the Chernoff bound}
\label{sec:asymptotics}

The various single-shot decoupling bounds, that we obtained in the previous sections for the multiple 
state discrimination problem, can be summarized as follows:
\begin{lemma}\label{lemma:singleshot}
For every $r\in\bN$, there exist $\kappa_r,\kappa'_r>0$ such that for all
$A_1,\ldots,A_r\in\B(\hil)_+$,
\begin{align}\label{mixed decoupling}
P_e^*(A_1,\ldots,A_r)\le\kappa_r(\Tr A_0)^{1/2}\sum_{(i,j):\,i\ne j}P_e^*(A_i,A_j)^{1/2}.
\end{align}
If all but at most two of the $A_i$ are of rank $1$ then we also have
\begin{align}\label{pure decoupling}
P_e^*(A_1,\ldots,A_r)\le\kappa_r'\frac{\Tr A_0}{\min_i\Tr A_i}\sum_{(i,j):\,i\ne j}P_e^*(A_i,A_j).
\end{align}
\end{lemma}
\begin{proof}
The bound in \eqref{mixed decoupling} can be obtained from either of the following: the bound \eqref{BK bound} of
\cite{barnumknill} using the Fuchs--van de Graaf inequalities; from the bound \eqref{eq:mixedboundT2} of Theorem \ref{th:half};
from the bound \eqref{mixed dich upper} of Theorem \ref{thm:dichotomic decoupling}; and from \eqref{nussbaum mixed} of Theorem \ref{th:nussbaum}.

The bound \eqref{pure decoupling} follows from \eqref{nussbaum pure} of Theorem \ref{th:nussbaum}.
(We can assume without loss of generality that at most hypotheses $1$ and $2$ are not represented by rank one operators.)
However,
when all the $A_i$ are pure, \eqref{pure decoupling} also follows from any of the following: from the bound \eqref{HLS bound} of \cite{HLS}
using the Fuchs--van de Graaf inequalities; from the bound \eqref{eq:pureboundT2} of Theorem \ref{th:half};
and from the bound \eqref{pure dich upper} of Theorem \ref{thm:dichotomic decoupling}.
\end{proof}
\smallskip

Armed with these upper bounds,
we now turn to the study of its asymptotic behavior.
Let our hypotheses be represented by the sequences $\vecc{A}_i:=\{A_{i,n}\}_{n\in\bN}$, $i=1,\ldots,r$, and define
$A_{0,n}:=\sum_{i=1}^r A_{i,n}$. Recall the definitions of
$\pli\bz\vecc{A}_1,\ldots,\vecc{A}_r\jz$ and $\pls\bz\vecc{A}_1,\ldots,\vecc{A}_r\jz$ from \eqref{li exponent}--\eqref{ls exponent}.
Due to Theorem \ref{th:41}, we have
\begin{equation}\label{asymptotic lower bound}
\pli\bz\vecc{A}_1,\ldots,\vecc{A}_r\jz
\ge
\max_{(i,j):\,i\ne j}\pli\bz\vecc{A}_i,\vecc{A}_j\jz.
\end{equation}
Our aim here is to complement the above inequality by giving upper bounds on $\pls\bz\vecc{A}_1,\ldots,\vecc{A}_r\jz$ in terms of the pairwise exponents.
Recall the definition of the asymptotic Chernoff divergence from \eqref{regularized Chernoff},
\begin{equation*}
C(\vecc{A}_1,\vecc{A_2})=\liminf_{n\to\infty}\frac{1}{n}\ch{A_{1,n}}{A_{2,n}}=
-\limsup_{n\to\infty}\frac{1}{n}\log \min_{0\le s\le 1}\Tr A_{1,n}^s A_{2,n}^{1-s}.
\end{equation*}

We conjecture that the following converse to \eqref{asymptotic lower bound} holds under very mild conditions:
\begin{align}\label{asymptotic upper conjecture}
\pls\bz\vecc{A}_1,\ldots,\vecc{A}_r\jz
\le
\max_{(i,j):\,i\ne j}\pls\bz\vecc{A}_i,\vecc{A}_j\jz
\le
-\min_{(i,j):\,i\ne j}C(\vecc{A}_1,\vecc{A_2}).
\end{align}
Note that the second inequality is always true, due to \eqref{Chernoff upper}. Below we show that the weaker inequality
\begin{align}\label{asymptotic upper half}
\pls\bz\vecc{A}_1,\ldots,\vecc{A}_r\jz
\le
\half\max_{(i,j):\,i\ne j}\pls\bz\vecc{A}_i,\vecc{A}_j\jz
\le
-\half\min_{(i,j):\,i\ne j}C(\vecc{A}_1,\vecc{A_2})
\end{align}
is always true as long as $\limsup_{n\to\infty}\frac{1}{n}\log\Tr A_{0,n}=0$, which is trivially satisfied in the case of weighted states. We also show
\eqref{asymptotic upper conjecture} in a number of special cases.

We have the following general result:
\begin{thm}\label{thm:general asymptotic upper bounds}
Assume that $\limsup_{n\to\infty}\frac{1}{n}\log\Tr A_{0,n}=0$. Then
\begin{align}\label{asymptotic upper}
\pls\bz\vecc{A}_1,\ldots,\vecc{A}_r\jz
\le
\half\max_{(i,j):\,i\ne j}\pls\bz\vecc{A}_i,\vecc{A}_j\jz
\le
-\half\min_{(i,j):\,i\ne j}C(\vecc{A}_1,\vecc{A_2}).
\end{align}
Assume, moreover, that $A_{i,n}$ is of rank one for every $n\in\bN$ for at least $r-2$ of the hypotheses. 
If $\limsup_{n\to\infty}\frac{1}{n}\log\frac{\Tr A_{0,n}}{\min_i\Tr A_{i,n}}=0$, then
we have the stronger inequality
\begin{align}\label{asymptotic upper pure}
\pls\bz\vecc{A}_1,\ldots,\vecc{A}_r\jz
\le
\max_{(i,j):\,i\ne j}\pls\bz\vecc{A}_i,\vecc{A}_j\jz
\le
-\min_{(i,j):\,i\ne j}C(\vecc{A}_i,\vecc{A_j}).
\end{align}
\end{thm}
\begin{proof}
Immediate from Lemma \ref{lemma:singleshot}, Lemma \ref{lemma:maximum rate}, and \eqref{Chernoff upper}.
\end{proof}

We can also prove \eqref{asymptotic upper conjecture} in the following special cases,
by using Theorem \ref{th:nussbaum}.


\begin{thm}\label{thm:recursive}
Assume that \eqref{asymptotic upper conjecture} holds for hypotheses $\vecc{A}_i,\,i=1,\ldots,r$, and that $A_{r+1,n}$ is rank one for every $n$.
If  $\limsup_{n\to\infty}\frac{1}{n}\log\frac{\Tr\sum_{i=1}^r A_i}{\Tr A_{r+1,n}}=0$ then
\begin{align*}
\pls\bz\vecc{A}_1,\ldots,\vecc{A}_r,\vecc{A}_{r+1}\jz
&\le
\max\left\{\pls\bz\vecc{A}_i,\vecc{A}_j\jz:\,1\le i<j\le r+1\right\}\\
&\le
-\min\left\{C(\vecc{A}_i,\vecc{A_j}):\,1\le i<j\le r+1\right\}.
\end{align*}
\end{thm}
\begin{proof}
By \eqref{Nussbaum bound},
\begin{align*}
P_e^*(A_{1,n},\ldots,A_{r,n}) \le 2P_e^*(A_1,\ldots,A_r) + P_e^*(3A^{(r)},A_{r+1}).
\end{align*}
Applying then \eqref{pure upper bounds} to the second term yields the assertion.
\end{proof}

\begin{rem}
Note that the binary case \eqref{Chernoff upper}, combined with a recursive application of Theorem \ref{thm:recursive}, 
gives an alternative proof of the second part of Theorem \ref{thm:general asymptotic upper bounds}.
\end{rem}

Inequality \eqref{asymptotic upper conjecture} has been proved in \cite{N13} for the i.i.d.~case under the assumption that 
there exists a pair of states $\sigma_k,\sigma_l,\,k\ne l$, such that $C(\sigma_k,\sigma_l)\le \frac{1}{6}C(\sigma_i,\sigma_j)$ 
for every $(i,j)\ne (k,l),\,i\ne j$. Theorem
\ref{thm:nussbaum achievability} below shows that the constant $1/6$ can be improved to $1/2$.

\begin{thm}\label{thm:nussbaum achievability}
Assume that $\limsup_{n\to\infty}\frac{1}{n}\log\Tr A_{0,n}=0$. For any pair $(k,l),\,k\ne l$,
\begin{align*}
\pls\bz\vecc{A}_1,\ldots,\vecc{A}_r\jz
&\le
\max\left\{\pls(\vecc{A}_k,\vecc{A}_l),\,
\half\pls(\vecc{A}_i,\vecc{A}_j),\,i\ne j,\,(i,j)\ne (k,l) \right\}\\
&\le
-\min\left\{C(\vecc{A}_k,\vecc{A}_l),\,\half C(\vecc{A}_i,\vecc{A}_j),\,i\ne j,\,(i,j)\ne (k,l)\right\}.
\end{align*}
In particular, if there exists a pair $(k,l),\,k\ne l$, such that $\pls(\vecc{A}_k,\vecc{A}_l)\ge\half\pls(\vecc{A}_i,\vecc{A}_j)$
or $C(\vecc{A}_k,\vecc{A}_l)\le\half C(\vecc{A}_i,\vecc{A}_j)$, $i\ne j,\,(i,j)\ne (k,l)$, then
\begin{align*}
\pls\bz\vecc{A}_1,\ldots,\vecc{A}_r\jz
\le
\max_{(i,j):\,i\ne j}\pls\bz\vecc{A}_i,\vecc{A}_j\jz
\le
-\min_{(i,j):\,i\ne j}C(\vecc{A}_i,\vecc{A_j}).
\end{align*}
\end{thm}
\begin{proof}
Immediate from Theorem \ref{th:nussbaum}.
\end{proof}
\medskip

Finally, we note that in many important cases, we have the optimality relation
\begin{equation}
\pli\bz\vecc{A}_i,\vecc{A}_j\jz\ge -C\bz\vecc{A}_i,\vecc{A}_j\jz.
\end{equation}	
For instance, this happens in the standard state discrimination problem if the hypotheses $i,j$ are i.i.d.~\cite{NSz}, or
Gibbs states of a finite-range, translation-invariant Hamiltonian on a spin chain \cite{HMO2}, or Gibbs states of 
interaction-free fermionic or bosonic
chains \cite{MHOF,M}. In these cases, if $\limsup_{n\to\infty}\frac{1}{n}\log\Tr A_{0,n}=0$ then we have
\begin{equation*}
-\min_{(i,j):\,i\ne j}C\bz\vecc{A}_i,\vecc{A}_j\jz\le
\pli\bz\vecc{A}_1,\ldots,\vecc{A}_r\jz
\le
\pls\bz\vecc{A}_1,\ldots,\vecc{A}_r\jz
\le
-\half\min_{(i,j):\,i\ne j}C\bz\vecc{A}_i,\vecc{A}_j\jz.
\end{equation*}
If, moreover, \eqref{asymptotic upper conjecture} is satisfied then we get the stronger statement
\begin{equation*}
\lim_{n\to+\infty}\frac{1}{n}\log P_e^*\bz A_{1,n},\ldots,A_{r,n}\jz
=
-\min_{(i,j):\,i\ne j}C\bz\vecc{A}_i,\vecc{A}_j\jz.
\end{equation*}

\appendix

\s\bigskip

\noindent\textbf{\LARGE Appendix}

\section{Least upper bound and greatest lower bound for operators}
\label{sec:LUB}

As mentioned already in Section \ref{sec:ML}, for a set $A_1,\ldots,A_r$ of self-adjoint operators on the same Hilbert space, the set of upper bounds
$\A:=\{Y:\,Y\ge A_k,\,k=1,\ldots,r\}$ has no minimal element in general. The following example shows that a minimal element 
may not exist even if all the $A_k$ commute with each other.

\begin{ex}\label{ex:no max}
Let $\hil=\bC^2$, and let the operators $A_1,A_2,Y_{\alpha,\beta,\delta}$ be given by their matrices in the standard basis of $\bC^2$ as
\begin{align*}
A_1:=\begin{bmatrix} 1 & 0 \\ 0 & 2\end{bmatrix},\ds\ds
A_2:=\begin{bmatrix} 2 & 0 \\ 0 & 1\end{bmatrix},\ds\ds
Y_{\alpha,\beta,\delta}:=\begin{bmatrix} 2+\alpha & \delta \\ \overline{\delta} & 2+\beta\end{bmatrix}.
\end{align*}
Let $\I:=\{(\alpha,\beta,\delta)\in\bR^3:\,\alpha,\beta\ge 0,\,
\min\{\alpha,\beta\}+\alpha\beta\ge|\delta|^2\}$.
It is easy to see that $\U(A_1,A_2)=\{Y:\,Y\ge A_1,A_2\}=\{Y_{\alpha,\beta,\delta}:\,(\alpha,\beta,\delta)\in\I\}$.
Assume that $\U(A_1,A_2)$ has a minimal element $Y$.
The assumption $Y\ge A_1,A_2$ yields that $Y_{11}\ge 2$ and $Y_{22}\ge 2$,
while the assumption that $Y\le Y_{\alpha,\beta,\delta}$ for all $(\alpha,\beta,\delta)\in\I$ yields that $Y_{11}\le 2$ and $Y_{22}\le 2$. Hence,
$0\le Y-A_1=\begin{bmatrix} 1 & Y_{12} \\ \overline{Y_{12}} & 0\end{bmatrix}$, which yields $Y_{12}=0$, i.e., $Y=2I$. Now,
$Y_{\alpha,\beta,\delta}-Y\ge 0$ if and only if $\alpha,\beta\ge 0$ and
$\alpha\beta\ge |\delta|^2$, which defines a strictly smaller set than $\I$, contradicting our initial assumption that $Y$ is a lower bound to $\U(A_1,A_2)$.
\end{ex}

In general, the set $\A:=\{Y:\,Y\ge A_k,\,k=1,\ldots,r\}$ is the intersection of $r$ cones,
and the intersection of two cones
is not itself a cone, unless one is completely contained in the other.
Thus, $\A$ has no unique minimal element in general, in the sense that there would be an element $Y_0$ such that $Y_0\le Y$ for all $Y\in\A$.
Rather, there is an infinity of minimal elements, in the sense that there is an infinity of operators $Y\in \A$ for
which no other $Y'\in\A$ exists such that $Y'\le Y$,
and these minima constitute the boundary of $\A$ \cite{ando93}.
The upshot is that one can not define a least upper bound on the basis of the PSD ordering alone.

However, there is a unique minimal element within $\A$ in terms of the \emph{trace ordering}.
We can therefore define a least upper bound in this more restrictive sense as
\be
\lub(A_1,\ldots,A_r) := \argmin_{Y}\{\trace Y: Y\ge A_k,\, k=1,\ldots,r\}.\label{eq:defLUB2}
\ee
To make sense of the definition, we have to prove the uniqueness of the minimizer.
For the proof, we will need the following simple fact, which has been stated, e.g., in
\cite{ando93} without a proof. Here we provide a proof for readers' convenience.
\begin{lemma}\label{lemma:supports}
Let $D,T\in\B(\hil)$ be self-adjoint operators such that
$D\ge \pm T$. Then $D$ is positive semidefinite, and its support dominates the support of $T$.
\end{lemma}
\begin{proof}
First, $D\ge \pm T$ implies $D\ge (T+(-T))/2=0$, proving that $D$ is PSD.
Let $\hil_1$ denote the support of $D$, and decompose $\hil$ as $\hil=\hil_1\oplus\hil_2$.
Then $D$ and $T$ can be written in the corresponding block forms as
$D=\begin{bmatrix} D_{11} & 0\\0 & 0\end{bmatrix}$ and
$T=\begin{bmatrix} T_{11} & T_{12}\\T_{12}^* & T{22}\end{bmatrix}$,
and positive semidefiniteness of $D\pm T$ implies $0\ge T_{22}\ge 0$. Using again that
$D+T\ge 0$, we finally obtain that $T_{12}=0$, too, from which the assertion follows.
\end{proof}

\begin{theorem}
Let $A_1,\ldots,A_r\in\B(\hil)\sa$ be a finite number of self-adjoint operators.
Then in the set $\A:=\{Y: Y\ge A_1,\ldots,A_r\}$ there is a unique element with minimal trace.
\end{theorem}
\begin{proof}
Let us assume that there are two distinct elements $Y_1$ and $Y_2$ in $\A$ with minimal trace $\trace Y_1=\trace Y_2$.
Let $Y_m=(Y_1+Y_2)/2$ and $\Delta=(Y_1-Y_2)/2$. Then $Y_1=Y_m+\Delta$ and $Y_2=Y_m-\Delta$,
and
$Y_1, Y_2\ge A_i$ implies $Y_m-A_i\ge\pm\Delta$. Hence, by Lemma \ref{lemma:supports},
there exists a constant $c_i>0$ such that $Y_m-A_i\ge c_i|\Delta|$
for every $i=1,\ldots,r$. Taking $c:=\min_i c_i$, we have
$Y_m-c|\Delta|\ge A_i,\,i=1,\ldots,r$. Thus, $Y_m-c|\Delta|\in\A$, but
$\Tr(Y_m-c|\Delta|)=\Tr Y_m-c\Tr|\Delta|<\Tr Y_m=\Tr Y_i,\,i=1,2$, contradicting our original assumption.
\end{proof}
\medskip

Next, we explore some properties of the LUB.
It is easy to see from (\ref{eq:defLUB2}) that the LUB satisfies the \emph{translation property}:
\be
\lub(A_1+B,\ldots,A_r+B) = \lub(A_1,\ldots,A_r)+B.\label{eq:LUBshift}
\ee
This is because the addition $X\mapsto X+B$, with a fixed self-adjoint operator $B$, is an order-preserving operation.
Furthermore, the LUB is jointly homogeneous: for any $c\ge0$,
\be
\lub(cA_1,\ldots,cA_r) = c\lub(A_1,\ldots,A_r).\label{eq:LUBhom}
\ee

The positive part and modulus can be expressed in terms of the LUB.
\begin{lemma}\label{lemma:positive part}
For all Hermitian operators $A$,
\be
A_+ = \lub(A,0), \ds\ds\mbox{ and }\ds\ds |A| = \lub(A,-A).
\ee
\end{lemma}
\begin{proof}
Consider the set $\A=\{Y: Y\ge A, Y\ge0\}$. Clearly, $A_+\in\A$.
By Weyl's monotonicity principle, the eigenvalues of any $Y\in \A$ are non-negative and not smaller than those of $A$; that is,
$\lambda_j(Y)\ge\lambda_j(A)$, where $\lambda_j$ denotes the $j^{th}$ largest eigenvalue.
Hence, $\lambda_j(Y)\ge\lambda_j(A_+)$, since the spectrum of $A_+$ consists of the positive eigenvalues of $A$ and zero.
As the sum of all eigenvalues is the trace,
$A_+$ is an element (and therefore \emph{the} element) in $\A$ with minimal trace.

Using (\ref{eq:LUBshift}) and (\ref{eq:LUBhom}), the modulus $|A| = 2A_+ - A$ can be similarly expressed as
$|A| = 2\lub(A,0) -A = \lub(A,-A)$.
\end{proof}

\begin{rem}
We emphasize again that the $\lub$ is a minimum with respect to the trace ordering and not the PSD ordering. In particular,
$X\ge A$ and $X\ge -A$ doesn't imply $X\ge |A|$. A counterexample can be easily given by taking $A=\begin{bmatrix} 1 & 0\\0 & -1\end{bmatrix}$
and $X=\begin{bmatrix} 2 & \sqrt{3}\\\sqrt{3} & 2\end{bmatrix}$. However, as lemma \ref{lemma:supports} shows, 
there always exists a positive constant $c$, depending on $A$ and $X$, such that
$X\ge c|A|$.
\end{rem}
\smallskip

By lemma \ref{lemma:positive part} and (\ref{eq:LUBshift}), $(A-B)_+ = \lub(A,B)-B$.
This immediately leads to a closed form expression for the LUB of two Hermitian operators:
\begin{lemma}
For all Hermitian operators $A,B$,
\be
\lub(A,B) = B+(A-B)_+ = \half(A+B+|A-B|) = A+(A-B)_-. \label{eq:LUB2}
\ee
\end{lemma}
\noindent From these expressions it is clear that for $A,B\ge0$, the LUB is PSD as well.
\medskip

In a similar vein we can define the \textit{greatest lower bound} (GLB) as
\be
\glb(A_1,\ldots,A_r) := \argmax_Y\{\trace Y: Y\le A_k,\, k=1,\ldots,r\}.\label{eq:defGLB2}
\ee
Clearly, we have
\be
\glb(A_1,\ldots,A_r)  = -\lub(-A_1,\ldots,-A_r).\label{eq:LUBGLB2}
\ee
Hence, for two operators, we get
\begin{lemma}
For all Hermitian operators $A,B$,
\be
\glb(A,B) = \half(A+B-|A-B|) = A-(A-B)_+ = B-(A-B)_-. \label{eq:GLB2}
\ee
\end{lemma}

A warning is in order about the sign of the GLB.
When $A$ and $B$ commute, their GLB is given by the entrywise minimum in the joint eigenbasis.
If $A$ and $B$ are also PSD, then clearly their GLB will be PSD.
When $A$ and $B$ are PSD but do not commute, however, their GLB need not be PSD;
only the trace of their GLB will be guaranteed to be non-negative. The reason is that while the function $x\mapsto x_+ = \max(0,x)$ is
monotone increasing, it is also convex and therefore not operator monotone. Thus, for $A,B\ge0$, $(A-B)_+\le A$ need not be true.
For a concrete counterexample, take $A=\pr{x},\,B=\pr{y}$ with $x=(1,1),\,y=(1,i)$; then it is easy to check that $0\nleq\lub(A,B)$.
Similarly, the LUB of two negative semidefinite operators need not be negative semidefinite.

Both LUB and GLB are monotonous in their arguments with respect to the PSD ordering.
\begin{lemma}\label{lem:LUBGLBmono}
For all Hermitian operators $\{A_i\}$ and $\{B_i\}$, if $A_i\le B_i$ then
\bea
\trace\lub(A_1,\ldots,A_r) &\le& \trace\lub(B_1,\ldots,B_r), \label{eq:LUBmono} \\
\trace\glb(A_1,\ldots,A_r) &\le& \trace\glb(B_1,\ldots,B_r). \label{eq:GLBmono}
\eea
\end{lemma}
\begin{proof}
By definition, $\lub(B_1,\ldots,B_r)\ge B_i\ge A_i$, for all $i$, so that $\lub(B_1,\ldots,B_r)$ is an upper bound
on all $A_i$. In general it is not the minimal one, hence $\trace\lub(B_1,\ldots,B_r)\ge\trace\lub(A_1,\ldots,A_r)$.
Monotonicity for the GLB follows from this and the correspondence (\ref{eq:LUBGLB2}).
\end{proof}

The LUB and GLB (and their trace) behave in the expected way
with respect to the direct sum:
\begin{lemma}\label{lub direct sum}
For any pair of sets of $A_i\in\B(\hil_1)\sa$ 
and $B_i\in\B(\hil_2)\sa$, 
$i=1,\ldots,r$,
\bea
\lub(\{A_i\oplus B_i\}) &=& \LUB(\{A_i\}) \oplus \LUB(\{B_i\}) \\
\glb(\{A_i\oplus B_i\}) &=& \GLB(\{A_i\}) \oplus \GLB(\{B_i\}).
\eea
\end{lemma}
\begin{proof}
Consider first the LUB.
Let $X:=\lub(\{A_i\oplus B_i\})$, and
let $P_i$ denote the projection onto $\hil_i$ in the direct sum $\hil_1\oplus\hil_2$.
Then $P_1XP_1\oplus P_2XP_2\ge A_i\oplus B_i$ for all $i$, and
$\Tr X=\Tr P_1XP_1\oplus P_2XP_2$. The uniqueness of the LUB then yields
$X=P_1XP_1\oplus P_2XP_2$.

The proof for the GLB goes exactly the same way.
\end{proof}

This lemma has an important consequence.
For every set of subnormalized states $\{A_i\}_{i=1}^r$ there is a set of normalized states $\{\sigma_i\}_{i=1}^r$
such that $\trace\GLB(\{A_i\})=\trace\GLB(\{\sigma_i\})$; namely
$\sigma_i=A_i \oplus (1-\trace A_i)|i\rangle\langle i|$, where $\{\ket{i}\}_{i=1}^r$ is an orthonormal system.
This is because the `appended' states $B_i=(1-\trace A_i)|i\rangle\langle i|$ are mutually orthogonal so that
$\trace\GLB(\{B_i\})=0$.
Similar statements can be made when the arguments
of $\trace\GLB$ are linear combinations of states.
The upshot of this is two-fold. First, for a large class of statements it allows one to restrict to normalized states to prove them.
Secondly, it aids the heuristic processes of coming up with reasonable conjectures and finding counterexamples (see, e.g., at the end of Section \ref{sec:averaged}).
\medskip

Finally, we give another representation of the least upper bound as the max-relative entropy center in the case where all the operators are positive semidefinite.
For PSD operators $A,B\in\B(\hil)_+$, their max-relative entropy $\dmax{A}{B}$ is defined as \cite{Datta,RennerPhD}
\begin{equation*}
\dmax{A}{B}:=\inf\{\gamma:\,A\le 2^{\gamma}B\}.
\end{equation*}
For a set of states $\A\subset\S(\hil)$, its max-relative entropy radius $R_{\max}(\A)$ is defined as
$R_{\max}(\A):=\inf_{\omega\in\S(\hil)}\sup_{\sigma\in\A}\dmax{\sigma}{\omega}$. For the interpretation of this 
quantity in quantum information theory, see, e.g.~\cite{KRS,MD,MH} and references therein. We extend this definition to 
general positive semidefinite operators by keeping the reference $\omega$ varying only over the set of states. That is, for a set of PSD operators
$\A\subset\B(\hil)_+$, its max-relative entropy radius $R_{\max}(\A)$ is defined as
\begin{equation*}
R_{\max}(\A):=\inf_{\omega\in\S(\hil)}\sup_{A\in\A}\dmax{A}{\omega}.
\end{equation*}
Any state $\omega$ where the infimum above is attained is called a \ki{$D_{\max}$-divergence center} of $\A$.

If $\A=\{0\}$ then $R_{\max}(\A)=-\infty$, and any state is a divergence center.
Assume for the rest that $\A=\{A_1,\ldots,A_r\}$ is finite, and it contains a non-zero element,
and hence $R:=R_{\max}(\A)$ is a finite number.
By definition, for every $n\in\bN$, there exists an $\omega_{n}\in\S(\hil)$ such that
$A\le 2^{R+1/n}\omega_{n}$ for every $A\in\A$. Since $\S(\hil)$ is compact, there exists
a subsequence $n_k,\,k\in\bN$, such that $\omega_{n_k},\,k\in\bN$, is convergent. Let $\omega^*:=\lim_{k\to\infty}\omega_{n_k}$; then
$A\le 2^R\omega^*$ for every $A\in\A$, and hence $\omega^*$ is a divergence center. Thus, the set of divergence centers is non-empty. 
Obviously, if $\omega$ is a divergence center then $2^R\omega$ is an upper bound to $\A$, and hence
$2^R\omega\ge\lub(\A)=:L$. Let $\tilde R:=\log_2\Tr L$ and $\tilde\omega:=L/\Tr L$. Then
$2^R\omega\ge L$ yields $R\ge\tilde R$, while
$A\le 2^{\tilde R}\tilde\omega$ due to the definition of $\lub(\A)$, and hence $R\le\tilde R$.
Thus, $R=\tilde R$, i.e., $\Tr(2^R\omega)=\Tr\lub(\A)$. Taking into account that $2^R\omega\ge L$, this implies that
$2^R\omega=\lub(\A)$. Thus, the $D_{\max}$-divergence center is unique, and is equal to
$\lub(\A)/\Tr\lub(\A)$, while $R_{\max}(\A)=\log\Tr\lub(\A)$.

According to \cite{ykl} (see also Appendix \ref{sec:semidef}), this can be rewritten as
$\log P_e^*(A_1,\ldots,A_r)=R_{\max}(\A)$. A similar expression for the optimal error probability in terms of
the max-relative entropy has been given in \cite{KRS}.

\section{The classical case}
\label{sec:classical}

In the classical case the hypotheses (in the single-shot setting) are represented by non-negative functions
$A_i:\,\X\to\bR_+$, where $\X$ is some finite set, and POVM elements are replaced by non-negative functions
$E_i:\,\X\to\bR_+$, satisfying $\sum_i E_i(x)\le 1,\,\forall x\in\X$, which we may call a classical POVM.
The success probability corresponding to a classical POVM $\{E_i\}$ is
$P_s\bz\{E_i\}\jz=\sum_{i=1}^r\sum_{x\in\X}A_i(x)E_i(x)$. We can assign to each non-negative function $F:\,\X\to\bR_+$ a diagonal
operator $\hat F$ on $\bC^{\X}$ in an obvious way, and under this identification we get
$\sum_{i=1}^r\sum_{x\in\X}A_i(x)E_i(x)=\sum_{i=1}^r\Tr \hat A_i\hat E_i$, which is the success probability corresponding to hypotheses $\hat A_i$ and POVM elements $\hat E_i$.
On the other hand, if $A_1,\ldots,A_r\in\B(\hil)_+$ are mutually commuting then there exists a basis in $\hil$, labeled by the elements of some finite set $\X$, such that
$A_i=\sum_{x\in\X}\pinner{x}{A_i}{x}\pr{x}$. Moreover, for any operator $E\in\B(\hil)$, we have
$\Tr A_iE=\sum_{x\in\X}\tilde A(x)\tilde E_(x)$, where for $F\in\B(\hil)$, we let $\tilde F:\,\X\to\bC$ be defined by $\tilde F(x):=\pinner{x}{F}{x}$. In particular, if
$E_1,\ldots,E_r$ is a POVM then $\tilde E_1,\ldots,\tilde E_r$ is a classical POVM, and
$\sum_{i=1}^r\sum_{x\in\X}\tilde A_i(x)\tilde E_i(x)=\sum_{i=1}^r\Tr A_iE_i$. Hence, if the operators representing the 
hypotheses are diagonal in a given basis then it is enough to consider POVM elements that are also diagonal in the same 
basis, which reduces the problem into a classical one. Thus, the classical case can be represented both by functions and 
diagonal operators, and we will not make a difference in the notation between the two representations in what follows.

Consider first the classical binary state discrimination problem with hypotheses $A_1=A$ and $A_2=\sum_{i=1}^r B_i$.
Then we have the following strengthening of Theorem \ref{thm:averaged upper bounds}:
\begin{align}\label{classical subadditivity}
P_e^*\bz A,\sum\nolimits_i B_i\jz\le \sum_i P_e^*(A,B_i).
\end{align}
Indeed,
\begin{align*}
P_e^*\bz A,\sum\nolimits_i B_i\jz&=
\half\Tr\bz A+\sum\nolimits_i B_i\jz-\half\norm{A-\sum\nolimits_i B_i}_1\\
&=
\half\sum_x\left(A(x)+\sum\nolimits_i B_i(x)-\left|A(x)+\sum\nolimits_i B_i(x)\right|\right)\\
&=
\half\sum_x f_x\bz\sum\nolimits_i B_i(x) \jz,
\end{align*}
where $f_x(t):=t+x-|t-x|=2\min\{t,x\}$. It is easy to see that $f_x$ is subadditive for every $x$, and hence the above can be continued as
\begin{align*}
P_e^*\bz A,\sum\nolimits_i B_i\jz&=
\half\sum_x f_x\bz\sum\nolimits_i B_i(x) \jz
\le
\half\sum_x\sum_i f_x(B_i(x))
=
\sum_i\half\sum_x f_x(B_i(x))\\
&=
\sum_i\bz \half\Tr\bz A+B_i\jz-\half\norm{A-B_i}_1\jz
=
\sum_i P_e^*(A,B_i).
\end{align*}
Combining this with Theorem \ref{th:dich}, we get
\begin{equation}\label{classical decoupling}
P_e^*\bz A_1,\ldots,A_r\jz\le P_e^*
=
\sum_{k=1}^r P_e^*\bz A_k,\sum\nolimits_{l\ne k}A_l\jz
\le
\sum_{(k,l):\,k\ne l}^r P_e^*(A_k,A_l),
\end{equation}
proving Conjecture \ref{claim:1} with $c=2(r-1)$. Below we give a more direct proof of this, without using
Theorem \ref{th:dich}.

Consider now the i.i.d.~vs.~averaged i.i.d.~problem as in Section \ref{sec:averaged}, 
with hypotheses $A_{1,n}=p\rho^{\otimes n}$ and $B_{i,n}=(1-p)q_i\sigma_i^{\otimes n}$, where
$p\in(0,1)$ and $q$ is a probability distribution.  Then we have
\begin{align*}
\sum_i q_i P_e^*\bz p\rho^{\otimes n},(1-p)\sigma_i^{\otimes n}\jz
&\le
P_e^*\bz p\rho^{\otimes n},(1-p)\sum\nolimits_i q_i\sigma_i^{\otimes n}\jz\\
&\le
\sum_i P_e^*\bz p\rho^{\otimes n},(1-p)q_i\sigma_i^{\otimes n}\jz
\le
\sum_i P_e^*\bz p\rho^{\otimes n},(1-p)\sigma_i^{\otimes n}\jz,
\end{align*}
where the first inequality is due to the convexity of the trace-norm, the second is due to the 
subadditivity relation \eqref{classical subadditivity}, and the last inequality is obvious
from the definition of the error probability. This yields immediately Conjecture \ref{con:averaged asymptotics2} in the classical case, i.e.,
\begin{align*}
\lim_{n\to\infty}\frac{1}{n}\log P_e^*\bz p\rho^{\otimes n},(1-p)\sum\nolimits_i q_i\sigma_i^{\otimes n}\jz=
-\min_{i}C(\rho,\sigma_i).
\end{align*}
\smallskip

Consider now the classical single-shot state discrimination problem with hypotheses $A_1,\ldots,$ $A_r:\,\X\to\bR_+$,
and let $m(x):=\max_k A_k(x)$.
We say that a POVM $\{E_k\}_{k=1}^r$ is a \ki{maximum likelihood} POVM if $E_k(x)=0$ when $A_k(x)<m(x)$, and for every $x\in\X$, $\sum_k E_k(x)=1$.
For any POVM $\{E_k\}_{k=1}^r$, we have
\begin{align*}
P_s(E_1,\ldots,E_r)=\sum_k\sum_x A_k(x)E_k(x)\le\sum_k\sum_x m(x)E_k(x)\le \sum_x m(x)=
\Tr\max\{A_1,\ldots,A_r\},
\end{align*}
where $\max\{A_1,\ldots,A_r\}:=\sum_x m(x)\pr{x}$. The above inequality holds with equality 
if and only if $\{E_k\}_{k=1}^r$ is a maximum likelihood POVM, and hence we have
\begin{align*}
P_s^*(A_1,\ldots,A_r)=\Tr\max\{A_1,\ldots,A_r\}.
\end{align*}

Now let $E_1,\ldots,E_r$ be a maximum likelihood measurement. Then
the individual error probabilities are, for each $k$,
\begin{align*}
P_{e,k}=\sum_x A_k(x)\bz 1-E_k(x)\jz=
\sum_{x:\,A_k(x)<m(x)}A_k(x)+\sum_{x:\,A_k(x)=m(x)}A_k(x)\bz 1-E_k(x)\jz.
\end{align*}
Obviously, if $A_k(x)<m(x)$ then there exists an $l\ne k$ such that $A_k(x)<A_l(x)$, and if
$A_k(x)=m(x)$ and $A_k(x)\bz 1-E_k(x)\jz>0$ then there exists and $l\ne k$ such that $A_k(x)=A_l(x)$.
Hence,
\begin{align*}
P_{e,k}\le
\sum_{l\ne k} \sum_{x:\,A_k(x)<A_l(x)}A_k(x)+\sum_{l\ne k}\sum_{x:\,A_k(x)=A_l(x)} A_k(x)
=
\sum_{l\ne k}\sum_{x:\,A_k(x)\le A_l(x)}A_k(x).
\end{align*}
Thus,
\begin{align*}
P_e^*(A_1,\ldots,A_r)&=
\sum_{k=1}^r P_{e,k}
\le
\sum_{k=1}^r\sum_{l\ne k}\sum_{x:\,A_k(x)\le A_l(x)}A_k(x)\\
&\le
\sum_{k=1}^r\sum_{l\ne k}\left[\sum_{x:\,A_k(x)\le A_l(x)}A_k(x)+
\sum_{x:\,A_k(x)> A_l(x)}A_l(x)
\right]\\
&=
\sum_{k=1}^r\sum_{l\ne k}\left[\half\Tr(A_k+A_l)-\half\norm{A_k-A_l}_1
\right]\\
&=
\sum_{(k,l):\,k\ne l}P_e^*(A_k,A_l),
\end{align*}
and we recover \eqref{classical decoupling}.

\section{The pure state case}
\label{sec:pure}

Let $A_1,A_2\in\B(\hil)_+$ be rank one operators; then we can write them as $A_i=\pr{x_i}=p_i\pr{\psi_i}=p_i\sigma_i$, where $p_i:=\Tr A_i$. 
Many of the divergence measures coincide in this case; indeed, it is easy to see that
\begin{align*}
|\inner{\psi_1}{\psi_2}|^2&=F(\sigma_1,\sigma_2)^2=Q_s(\sigma_1\|\sigma_2)=Q_{\min}(\sigma_1,\sigma_2)=\exp(-C(\sigma_1,\sigma_2)),\ds\ds\ds s\in[0,1].
\end{align*}
A straightforward computation gives that $\norm{A_1-A_2}_1=\sqrt{(p_1+p_2)^2-4p_1p_2|\inner{\psi_1}{\psi_2}|^2}$, and hence
\begin{align*}
P_e^*(A_1,A_2)=\half\Tr(A_1+A_2)-\half\norm{A_1-A_2}_1=
\frac{2p_1p_2|\inner{\psi_1}{\psi_2}|^2}{p_1+p_2+\norm{A_1-A_2}_1}
\end{align*}
Noting that $0\le\norm{A_1-A_2}_1\le p_1+p_2$, we get
\begin{align}\label{pure two-sided bounds}
\frac{p_1p_2}{p_1+p_2}|\inner{\psi_1}{\psi_2}|^2\le P_e^*(A_1,A_2)\le\frac{2p_1p_2}{p_1+p_2}|\inner{\psi_1}{\psi_2}|^2
\end{align}
Consider now two sequences of rank one operators $\vecc{A_i}=\{A_{i,n}\}_{n\in\bN}$, $i=1,2$, and let $p_{i,n}:=\Tr A_{i,n}$, and
$A_{i,n}=p_{i,n}\pr{\psi_{i,n}}=p_{i,n}\sigma_{i,n}$. Applying \eqref{pure two-sided bounds} to each $n$, we get
\begin{align*}
\frac{p_{1,n}p_{2,n}}{p_{1,n}+p_{2,n}}\exp\bz-C(\sigma_{1,n},\sigma_{2,n}\jz
\le
P_e^*(A_{1,n},A_{2,n})
\le
\frac{2p_{1,n}p_{2,n}}{p_{1,n}+p_{2,n}}\exp\bz-C(\sigma_{1,n},\sigma_{2,n}\jz.
\end{align*}
If we assume now that $0<\liminf_n p_{i,n}\le\limsup_n p_{i,n}<+\infty$, $i=1,2$, then taking the limit $n\to\infty$ in the above formula yields
\begin{align*}
\lim_{n\to\infty}\frac{1}{n}\log P_e^*(A_{1,n},A_{2,n})=-C\bz\vecc{\sigma_1},\vecc{\sigma_2}\jz=-C\bz\vecc{A}_1,\vecc{A}_2\jz,
\end{align*}
where the last identity is straightforward to verify.
Thus in the pure state case we can get the Chernoff bound theorem from the above elementary argument, 
without using the trace inequality of \cite{Aud} or the reduction to classical states
from \cite{NSz}.

Consider now the case $r>2$, and let $A_1,\ldots,A_r\in\B(\hil)_+$ be rank one operators. Let $E_i:=A_0^{-1/2}A_iA_0^{-1/2}$ 
be the POVM elements of the pretty good measurement, where
$A_0:=\sum_{i=1}^r A_i$. It was shown in Appendix A of \cite{HLS} that for every $i$,
\begin{align*}
\Tr A_i(I-E_i)\le\frac{1}{p_i}\sum_{j:\,j\ne i}|\inner{x_i}{x_j}|^2.
\end{align*}
Summing it over $i$, we get
\begin{align}\label{pure upper decoupling}
P_e^*(A_1,\ldots,A_r)&\le P_e(\{E_1,\ldots,E_r\})\le
\sum_{i=1}^r\frac{1}{p_i}\sum_{j:\,j\ne i}|\inner{x_i}{x_j}|^2
\le
\frac{1}{\min_i p_i}\sum_{(i,j):\,i\ne j}\exp(-C(A_i,A_j)),
\end{align}
while Theorem \ref{th:41} yields
\begin{align}\label{pure lower decoupling}
P_e^*(A_1,\ldots,A_r)&\ge \frac{1}{r-1} \sum_{(k,l):\,k< l} P_e^*(A_k,A_l)
\ge
\frac{1}{r-1} \sum_{(k,l):\,k< l}\frac{1}{p_k+p_l}\exp(-C(A_k,A_l)),
\end{align}
where the last inequality is due to \eqref{pure two-sided bounds}. Note that \eqref{pure upper decoupling} 
also yields a decoupling bound for the error probabilities, as
$|\inner{x_i}{x_j}|^2/p_i\le (1+p_j/p_i)P_e^*(A_i,A_j)$ by \eqref{pure two-sided bounds}, and hence
\begin{align*}
P_e^*(A_1,\ldots,A_r)&\le
\sum_{i=1}^r\frac{1}{p_i}\sum_{j:\,j\ne i}|\inner{x_i}{x_j}|^2
\le
\frac{\Tr A_0}{\min_i p_i}\sum_{(i,j):\,i\ne j}P_e^*(A_i,A_j).
\end{align*}

Consider now the asymptotic case, with hypotheses $\vecc{A_i},\,i=1,\ldots,r$, and assume as before that
$0<\liminf_n p_{i,n}\le\limsup_n p_{i,n}<+\infty,\,\forall i$. Applying
\eqref{pure upper decoupling} and \eqref{pure lower decoupling} to every $n$, and taking the limit $n\to\infty$, we get
\begin{equation*}
\lim_{n\to\infty}\frac{1}{n}\log P_e^*(A_{1,n},\ldots,A_{r,n})=-\max_{(i,j):\,i\ne j}C\bz\vecc{A}_i,\vecc{A}_j\jz.
\end{equation*}

\section{Semidefinite program representations of success and error probabilities}
\label{sec:semidef}

The average success probability of a POVM $\{E_k\}$ for discriminating between $r$ PSD operators $\{A_k\}_{k=1}^r$ is given by
\be
P_{s}(\{E_k\}) = \sum_{k=1}^r \trace(A_k E_k),
\ee
and the optimal success probability $P_s^*$ is the maximum over all POVMs:
\be
P_s^* = \max\left\{P_{s}(\{E_k\}):\,\{E_k\}_{k=1}^r\s\text{ POVM}\right\}.\label{eq:primal}
\ee
In this section we consider the consequences of the following simple observation \cite{ykl}: in (\ref{eq:primal})
the maximum of a linear functional is taken over the set of POVMs,
which is a convex set.
This optimization problem is therefore a so-called \textit{semidefinite program} (SDP) \cite{boyd}.
One consequence is that $P_s^*$ can be efficiently calculated numerically by SDP solvers
even when no closed form analytical solution exists. Another, theoretically important consequence is that
the duality theory of SDPs allows to express the value of $P_s^*$ in a dual way as a minimization problem
\cite{eldar,jezek}.

The Lagrangian of problem (\ref{eq:primal}) is
\beas
{\cal L} &=& \sum_k\trace(A_k E_k) +\sum_k\trace(Z_k E_k)+\trace Y\bz I-\sum_k E_k\jz \\
&=& \trace Y + \sum_k\trace E_k(A_k+Z_k-Y),
\eeas
where the operators $Z_k$ and $Y$ are the Lagrange multipliers of the problem.
If the $Z_k$ are taken to be PSD, we see that always $P_s(\{E_k\}) \le {\cal L}$.
This does not change when maximizing over all POVMs, and certainly not when in the
maximization of ${\cal L}$ over the $E_k$ the POVM constraints are dropped.
Hence $P_s^*\le \max_{E_k}{\cal L}$.
This unconstrained maximization is easy to do; when $Y=A_k+Z_k$ for all $k$,
$\max_{E_k}{\cal L} = \trace Y$, otherwise it is positive infinity.
Minimizing this upper bound over all PSD $Z_k$ and all $Y$ yields the best upper bound on $P_s^*$.
The positivity condition on the $Z_k$ can be replaced by requiring
that for all $k$, $Y\ge A_k$. Minimizing over such $Y$ then gives
\begin{equation}\label{dual upper bound}
P_s^* \le \min_Y\{\trace Y:\, Y\ge A_k,\, k=1\ldots,r\},
\end{equation}
which is again an SDP, called the \textit{dual} of the original (\textit{primal}) SDP
(see, e.g., \cite{ykl} or \cite{eldar}, equations (15) and (16)).

Therefore, the optimal success probability is bounded above by the trace of the LUB of all weighted density operators:
\be
P_s^* \le \trace \lub(A_1,\ldots,A_r). \label{ineq:dual}
\ee
Note that in the classical case (all $A_i$ are diagonal, with diagonal elements $A_i(j)=p_i q_i(j)$)
the LUB is the entrywise
maximum, so that the dual SDP reproduces the maximum-likelihood formula
$P_s^* = \sum_j\max_i(A_i(j))$.

The difference between the maximum of the primal SDP ($P_s^*$)
and the minimum of the dual SDP is called the \textit{duality gap}.
One can show that the duality gap is zero, provided some mild technical conditions are satisfied
(e.g.\ Slater's conditions), in which case equality holds:
\be
P_s^* = \trace\lub(A_1,\ldots,A_r). \label{eq:dual}
\ee

If the duality gap is zero, then
the optimal $Z_k$ (denoted by $Z_k^*$) and the optimal POVM $\{E_k^*\}$ must necessarily satisfy a simple relation, called the
\textit{complementary slackness} condition.
Indeed, Let $Y^*$ be the operator where the minimum on the RHS of \eqref{dual upper bound} is attained.
As $\trace Y^*=\trace(Y^*\sum_k E_k^*)$
and $Y^*-A_k=Z_k^*$,
the equality $\sum_k\trace A_k E_k^* = \trace Y^*$ implies
$\sum_k \trace(Z_k^*E_k^*)=0$.
Since all $Z_k^*$ and $E_k^*$ are required to be PSD,
this actually means that
\be
Z_k^* E_k^*=0,\ds\ds \forall k. \label{eq:slack}
\ee

A simple consequence of these complementary slackness conditions is obtained by summing over $k$:
$\sum_k Z_k^* E_k^* = 0$. Noting that $Z_k=Y-A_k$, this yields
\be
Y^* = \sum_k A_k E_k^*. \label{eq:YKL}
\ee
Combined with the conditions $Y^*\ge A_k$ for all $k$,
these are the optimality conditions first obtained by Yuen, Kennedy and Lax \cite{ykl}.

\section{Short proofs of Barnum and Knill's and Tyson's bounds}
\label{sec:Tysonproofs}

\noindent\textit{Proof of Theorem \ref{th:BKbound}.}
The main ingredient of the proof is the following lemma (a slight improvement over Lemma 5 in \cite{barnumknill}, which lacked the factor $\half$).
Let $M$ be a positive semidefinite $n\times n$ matrix, symmetrically partitioned as the $2\times 2$ block matrix $M=\twomat{X}{Y}{Y^*}{Z}$, where $X$ is $n_1\times n_1$,
$Y$ is $n_1\times n_2$ and $Z$ is $n_2\times n_2$ (with $n=n_1+n_2$). 
Let $M^2$ be partitioned conformally.
Then the off-diagonal blocks of $M$ and $M^2$ satisfy
\[
||M_{1,2}||_2^2 \le \half||(M^2)_{1,2}||_1.
\]
Note that the validity of this lemma does not extend to general $m\times m$ partitions.

\begin{quote}
\noindent\textit{Proof of lemma.} We have $M_{1,2}=Y$ and $(M^2)_{1,2} = XY+YZ$.
Let us, without loss of generality, assume that $n_1\le n_2$.
From the singular value decomposition of $Y$ we can obtain a basis for representing $M$ in which $Y$ is pseudo-diagonal with non-negative diagonal elements.
Let (for $i=1,\ldots,n_1$) $x_i$ and $y_i$ be the diagonal elements of $X$ and $Y$, and $z_i$ the first $n_1$ diagonal elements of $Z$, all of which are non-negative.
As $M$ is PSD, any of its principal submatrices is PSD too, and we have $y_i\le \sqrt{x_iz_i}\le (x_i+z_i)/2$. Thus
\[
||Y||_2^2 = \sum_{i=1}^{n_1} y_i^2 \le \half\sum_{i=1}^{n_1} x_iy_i+y_i z_i =\half\sum_{i=1}^{n_1} (XY+YZ)_{i,i} \le\half ||XY+YZ||_1,
\]
as required. 
The last inequality follows from the inequality $|\trace A|\le ||A||_1$ applied to the square matrix obtained by padding  $XY+YZ$ with extra rows containing zero (an
operation that does not affect the trace norm).
\qed
\end{quote}

To prove Theorem \ref{th:BKbound}, let $X$ be the $r\times 1$ column matrix $X:=(A_j^{1/2}A_0^{-1/4})_{j=1}^r$.
Then $X^*X=\sum_{j=1}^r A_0^{-1/4} A_j A_0^{-1/4} = A_0^{1/2}$.
Furthermore, let $N=XX^*$. Then $N_{i,j}=A_i^{1/2} A_0^{-1/2} A_j^{1/2}$
and $(N^2)_{i,j}=(XX^*XX^*)_{i,j}=A_i^{1/2} A_j^{1/2}$.

For each value of $i=1,\ldots,r$ we now apply the lemma to the $2\times 2$ block matrix $M=\twomat{X}{Y}{Y^*}{Z}$ where
$X=N_{i,i}$,
$Y$ is the $i$-th row of $N$, but with the $i$-th column removed, 
and $Z$ is the submatrix of $N$ with the $i$-th row and $i$-th column removed.
Thus, $M_{1,2}=Y$ is itself a row block matrix consisting of the $r-1$ blocks $A_i^{1/2} A_0^{-1/2} A_j^{1/2}$ for fixed $i$ and $j\neq i$.
Likewise, $(M^2)_{1,2}$ is a row block matrix consisting of the $r-1$ blocks $A_i^{1/2} A_j^{1/2}$.
The lemma then implies, for all $i$,
\beas
\sum_{j: j\neq i} \trace A_i A_0^{-1/2}A_j A_0^{-1/2} &=& \sum_{j: j\neq i} ||A_i^{1/2} A_0^{-1/2} A_j^{1/2}||_2^2 = ||M_{1,2}||_2^2 \\
&\le& \half||(M^2)_{1,2}||_1 = \half||(A_i^{1/2} A_j^{1/2})_{j\neq i}||_1 \\
&\le& \half \sum_{j: j\neq i} ||A_i^{1/2} A_j^{1/2}||_1 = \half \sum_{j: j\neq i} F(A_i,A_j).
\eeas
The last inequality is just the triangle inequality for the trace norm.
Summing over all $i$ yields the stated bound on the error probability $P_e^{PG}$.
\qed

\bigskip

\noindent\textit{Proof of Theorem \ref{th:Tyson}.}
For any operator $X$ with $\norm{X}\le 1$ and any quantum state $\sigma$ we have
\begin{equation}\label{Tyson1}
1-\norm{X\sigma}_1 \le 1-\trace(X^*X\sigma) \le 1-\norm{X\sigma}_1^2 \le  2(1-\norm{X\sigma}_1).
\end{equation}
The first two inequalities both follow from H\"older's inequality (\cite{bhatia}, Cor IV.2.6):
$$
\trace (X^*X\sigma)\le \norm{X^*X\sigma}_1\le \norm{X\sigma}_1\;\norm{X} \le \norm{X\sigma}_1,
$$
and
$$
\norm{X\sigma}_1^2 = \norm{(X\sigma^{1/2})\sigma^{1/2}}_1^2
\le \norm{X\sigma^{1/2}}_2^2\;\;\norm{\sigma^{1/2}}_2^2 = \trace X^*X\sigma
$$
and the last inequality in \eqref{Tyson1} follows from $1-x^2\le2(1-x),\,x\in\bR$.
Applying \eqref{Tyson1} for $\sigma_k:=A_k/\Tr A_k$ and $X_k$ and summing over $k$ yields
$$
\sum_{k=1}^r (\Tr A_k)\bz 1-\norm{X_k\sigma_k}_1\jz
\le
\sum_{k=1}^r (\Tr A_k)\bz 1-\Tr (X_k^*X_k\sigma_k)\jz
\le
2\sum_{k=1}^r (\Tr A_k)\bz 1-\norm{X_k\sigma_k}_1\jz.
$$
Taking the minimum over all $X_k$ then yields the inequalities of the theorem.
\qed
\bigskip

\noindent\textit{Proof of Theorem \ref{th:Tyson2}.}
In the following we will abbreviate the expression $\sum_{k=1}^r A_k^2$ by $S$.

First note that the operator $\bigoplus_k X_k A_k$ is a pinching of the block operator
$(X_j A_k)_{j,k}$. This operator is the product of the column block operator ${\cal X}:=(X_j)_{j,1}$
and the row block operator ${\cal A}:=(A_k)_{1,k}$.
Because unitarily invariant norms do not increase under pinchings, we get
$$
\sum_k \norm{X_k A_k}_1
= \norm{\bigoplus_k X_k A_k}_1
\le \norm{(X_j A_k)_{j,k}}_1 =\norm{{\cal X}{\cal A}}_1
=\trace\left({\cal A}^*{\cal X}^*{\cal X}{\cal A}\right)^{1/2}.
$$
Noting that ${\cal X}^*{\cal X} = \sum_k X_k^* X_k=I$ yields
$$
\trace\left({\cal A}^*{\cal X}^*{\cal X}{\cal A}\right)^{1/2}
= \trace\left({\cal A}^*{\cal A}\right)^{1/2}
= \norm{\cal A}_1
= \norm{{\cal A}^*}_1
= \trace S^{1/2}.
$$
Hence,
$$
\sum_k \norm{X_k A_k}_1 \le \trace S^{1/2}.
$$
Equality can be achieved by taking the SQ measurement, $X_k=A_k S^{-1/2}$. Indeed,
from $X_k A_k = A_k S^{-1/2} A_k\ge0$,
we get
$$
\sum_k\norm{X_k A_k}_1=\sum_k\trace(X_kA_k)=\trace\left(S^{-1/2} \sum_k A_k^2\right)
= \trace S^{1/2}.
$$
This shows that the maximum of $\sum_k \norm{X_kA_k}_1$ over any complete set
of $r$ measurement operators $\{X_k\}$
is achieved for $X_k=A_k S^{-1/2}$ and is given by $\trace S^{1/2}$.
Hence, $\Gamma^* = \Tr A_0-\trace S^{1/2}$, which is (\ref{T2}).
\qed
\section*{Acknowledgments}
The authors are grateful to Giulio Chiribella for pointing out the decoupling bounds of \cite{barnumknill}.
This work was partly supported by an Odysseus grant from the Flemish FWO (KA)
and by the Marie Curie International Incoming Fellowship ``QUANTSTAT'' within
the 7th European Community Framework Programme (MM).
MM also acknowledges support by the European Research Council (Advanced Grant ``IRQUAT'').
The authors are grateful to the following institutions for their hospitality: the Fields Institute in Toronto (KA and MM), during the
Thematic Program on Mathematics in Quantum Information, 2009, where this research has started, and
the University of Ulm (KA).

\end{document}